\documentclass[11pt]{article}
\usepackage[utf8]{inputenc}
\usepackage[T1]{fontenc}
\usepackage[pdftex]{graphicx}

\usepackage{sidecap}
\usepackage{enumitem}

\usepackage{amssymb}
\usepackage{amsmath}
\usepackage{float}
\usepackage{morefloats}
\usepackage{graphicx}
\usepackage{rotating}
\usepackage{wrapfig}

\usepackage{cite} 

\usepackage{hyperref}

\usepackage{url}
\urldef{\maileva}\path|roden@di.ku.dk|
\urldef{\mailjacob}\path|jh@poplar.dk|

\usepackage{todonotes}
\DeclareGraphicsExtensions{.pdf,.png,.jpg,.jpeg,.mps}

\newcommand{\Oo}{\mathcal O} 
 
\newcommand{\meet}{\operatorname{m}} 
\newcommand{\Patrascu}{P{\v a}tra{\c s}cu} 
\newcommand{\tto}{\cdots}

\newcommand{\nil}{\mathbf{nil}}
\newcommand{\slfrac}[2]{\left.#1\middle/#2\right.}
\newcommand{\dual}[1]{\ensuremath{{#1}^*}}
\newcommand{\cotree}[1]{\dual{\overline{#1}}}
\newcommand{\linkablequery}{\operatorname{linkable}}
\newcommand{\oneflipquery}{\operatorname{one-flip-linkable}}
\newcommand{\insertrm}{\operatorname{insert}}
\newcommand{\remove}{\operatorname{remove}}

\newcommand{\vertexjoin}{\operatorname{join}}
\newcommand{\vertexsplit}{\operatorname{cut}}
\newcommand{\flip}{\operatorname{flip}}
\newcommand{\clustermerge}{\operatorname{merge}}
\newcommand{\clustersplit}{\operatorname{split}}
\newcommand{\expose}{\operatorname{expose}}
\newcommand{\EET}{\operatorname{EET}}

\newcommand{\set}[1]{\left\{#1\right\}}

\newcommand{\simplepath}[2]{#1\cdots{}#2}

\usepackage{amsthm}
\theoremstyle{plain}
\newtheorem{theorem}{Theorem}
\newtheorem{observation}{Observation}
\newtheorem{lemma}{Lemma}
\theoremstyle{definition}
\newtheorem{definition}{Definition}

\title{Dynamic Planar Embeddings of Dynamic Graphs}

\author{Jacob Holm \& Eva Rotenberg\\
    \small
    Department of Computer Science, University of Copenhagen, Denmark\\
    \texttt{\{jaho,roden\}@di.ku.dk}}

\begin{document}

\maketitle

\begin{abstract}
	We present an algorithm to support the dynamic embedding in the plane
	of a dynamic graph. An edge can be inserted across a face between
	two vertices on the face boundary (we call such a vertex pair linkable), and edges can be deleted. The planar
	embedding can also be changed locally by flipping components that are
	connected to the rest of the graph by at most two vertices. 
	Given vertices $u,v$, $\operatorname{linkable}(u,v)$
	decides whether $u$ and $v$ are linkable in the current embedding, and if so, returns a list of suggestions for the placement of $(u,v)$ in the embedding.
	For non-linkable vertices $u,v$, we define a new query, $\operatorname{one-flip-linkable}(u,v)$ providing a suggestion for a flip that will make them linkable if one exists.
	We support all updates and queries in $\Oo(\log^2 n)$ time. Our time
	bounds match those of Italiano et al. for a static (flipless) embedding of a dynamic graph.
	Our new algorithm is simpler, exploiting that
	the complement of a spanning tree of a connected plane graph is a
	spanning tree of the dual graph.  The primal and dual trees are
	interpreted as having the same Euler tour, and a main idea of the new algorithm is an elegant 
	interaction between top trees over the two trees via
	their common Euler tour.
\end{abstract}\newpage



\section{Introduction}

We present a data structure for supporting and maintaining a dynamic
planar embedding of a dynamic graph.
In this article, a \emph{dynamic graph} is a graph where edges can be
removed or inserted, and vertices can be cut or joined, but where an edge $(u,v)$ can only be added if it does not violate planarity. More precisely,
the edges around each vertex are ordered cyclically by the embedding, 
similar to the edge-list representation of the graph.
A \emph{corner} (of a face) is the gap between consecutive edges incident to some
vertex. Given two corners $c_u$ and $c_v$ of the same face $f$, 
incident to the vertices $u$ and $v$ respectively, the operation
$\insertrm(c_u,c_v)$ inserts an edge between $u$ and $v$ in the
dynamic graph, and embeds it across $f$ via the specified corners.  We provide
an operation $\linkablequery(u,v)$ that returns such a pair of corners 
$c_u$ and $c_v$ if they exist. If there are more options, we
can list them in constant time per option after the first.  A vertex may
be cut through two corners, and linkable vertices may be joined by
corners incident to the same face, if they are connected, or incident
to any face otherwise. That is, joining vertices corresponds to
linking them across a face with some edge $e$, and then contracting
$e$.

It may often be relevant to change the embedding, e.g. in order to be
able to insert an edge.  In a \emph{dynamic embedding}, the user is
allowed to change the embedding by what we call flips, that is, to
turn part of the graph upside down in the embedding.
Of course, the relevance of this depends on what we want to describe
with a dynamic plane graph.  If the application is to describe roads on
the ground, flipping orientation would not make much sense.  But if we
have the application of graph drawing or chip design in mind, flips
are indeed relevant.  In the case of chip design, a layer of a chip is
a planar embedded circuit, which can be thought of as a planar
embedded graph.  An operation similar to flip is also supported by most
drawing software.

Given two vertices $u,v$, we may ask whether they can be linked after
modifying the embedding with only one flip. We introduce a new operation, the $\oneflipquery(u,v)$
query, which answers that question, and returns the vertices and corners
describing the flip if it exists.

Our data structure is an extension to a well-known duality-based
dynamic representation of a planar embedded graph known as a
tree/cotree decomposition~\cite{DBLP:conf/soda/Eppstein03}.  We maintain
top trees~\cite{DBLP:journals/talg/AlstrupHLT05} both for the primal and dual spanning
trees. We use the fact that they share a common (extended) Euler tour - in a new way -
to coordinate the updates and enable queries that either
tree could not answer by itself.  All updates and queries in the combined
structure are supported in time $\Oo(\log^2 n)$, plus, in case of $\linkablequery(u,v)$, the length of the returned list.

\subsection{Dynamic Decision Support Systems}

An interesting and related problem is that of dynamic planarity
testing of graphs.  That is, we have a planar graph, we insert some
edge, and ask if the graph is still planar, that is, if there exist an embedding of it in the plane? 

The problem of dynamic planarity testing appears technically harder
than our problem, and in its basic form it is only relevant when the
user is completely indifferent to the actual embedding of the graph.
What we provide here falls more in the category of a decision support
system for the common situation where the desired solution is not
completely captured by the simple mathematical objective, in this case
planarity. We are supporting the user in finding a good embedding,
e.g., informing the user what the options are for inserting an edge
(the $\linkablequery$ query), but leaving the actual choice to the
user.  We also support the user in changing their mind about the
embedding, e.g. by flipping components, so as to make edge insertions
possible.  Using the $\oneflipquery$ query we can even suggest a flip
that would make a desired edge insertion possible if one exists.

\subsection{Previous work}
Dynamic graphs have been studied for several decades.  Usually, a \emph{fully
dynamic} graph is a graph that may be updated by the deletion or
insertion of an edge, while \emph{decremental} or \emph{incremental} refer to
graphs where edges may only be deleted or inserted, respectively.  A
dynamic graph can also be one where vertices may be deleted along with
all their incident edges, or some combination of edge- and vertex
updates~\cite{Patrascu:2007}. 

\nocite{Karger94randomsampling}

Hopcroft and Tarjan~\cite{Hopcroft1974} were the first to solve
planarity testing of static graphs in linear time.  
Incremental planarity testing was solved by La
Poutr\'{e}~\cite{LaPoutre:1994}, who improved on work by Di Battista,
Tamassia, and Westbrook~\cite{di1989incremental,di1996line,westbrook1992}, to obtain a total running time of $\Oo(\alpha(q,n))$
where $q$ is the number of operations, and where $\alpha$ is the inverse-Ackermann function. Galil, Italiano, and
Sarnak~\cite{Galil99fullydynamic} made a data structure for fully
dynamic planarity testing with $\Oo(n^{\slfrac{2}{3}})$ worst case
time per update, which was improved to $\Oo(n^{\slfrac{1}{2}})$ by
Eppstein et al.~\cite{Eppstein:1996}. For maintaining embeddings of planar
graphs, Italiano, La Poutr\'{e}, and Rauch~\cite{DBLP:conf/esa/ItalianoPR93}
present a data structure for maintaining a planar embedding of a
dynamic graph, with time $\Oo(\log^2 n)$ for update and for linkable-query, where
insert only works when compatible with the
embedding. 
The dynamic tree/cotree decomposition was first introduced by Eppstein et al.~\cite{EppItaTam-Algs-92} who used it to maintain the Minimum Spanning Tree (MST) of a planar embedded dynamic graph subject to a sequence of change-weight$(e,\Delta x)$ operations in $\Oo(\log n)$ time per update. 
Eppstein~\cite{DBLP:conf/soda/Eppstein03} presents a data
structure for maintaining the MST of a 
dynamic graph, which handles updates in $\Oo(\log
n)$ time if the graph remains plane. More precisely
the user specifies a combinatorial embedding in terms of a cyclic ordering of the
edges around each vertex. Planarity of the user specified embedding is
checked using Euler's formula. Like our algorithm, Eppstein's supports flips. The fundamental difference is that
Eppstein does not offer any support for keeping the embedding planar, e.g.,
to answer $\linkablequery(u,v)$, the user would in principle have
to try all $n^2$ possible corner pairs $c_u$ and $c_v$ incident to $u$ and $v$, and ask if $\insertrm(c_u,c_v)$ violates planarity.

As far as we know, the $\oneflipquery$ query has not been studied
before. Technically it is the most challenging operation supported in
this paper.

The highest lower bound for the problem of planarity testing is \Patrascu's $\Omega(\log n)$ lower bound for fully dynamic
planarity testing~\cite{Patrascu:2006}. From the reduction it is clear that this lower bound holds as well for maintaining an embedding of a planar graph as we do in this article.

\subsection{Results}
We present a data structure for maintaining an embedding of a planar graph under the following updates:
\begin{description}
	\item[delete($e$)] Deletes the specified edge, $e$.
	\item[insert($c_u$,$c_v$)] Inserts an edge through the specified corners, if compatible with the current embedding.
	\item[linkable($u$,$v$)] Returns a data structure representing the (possibly empty) list of faces that $u$ and $v$ have in common, and for each common face, $f$, the (nonempty) lists of corners incident to $u$ and $v$, respectively, on $f$.
	\item[articulation-flip($c_1$,$c_2$,$c_3$)] If the corners $c_1,c_2,c_3$ are all incident to the same vertex, and $c_1,c_2$ are incident to the same face, this update flips the subgraph between $c_1$ and $c_2$ to the space indicated by $c_3$.
	\item[separation-flip($c_1$, $c_2$, $c_3$, $c_4$)] If the corners $c_1,c_2$ are incident to the vertex $v$, and $c_3,c_4$ are incident to the vertex $u$, and $c_2,c_3$ are incident to the face $f$, and $c_1,c_4$ are incident to the face $g$. Flips the subgraph delimited by these corners, or, equivalently, by $u,g,v,f$.
	\item[one-flip-linkable($u$,$v$)] Returns a suggestion for an articulation-flip or separation-flip making $u$ and $v$ linkable, if they are not already.
\end{description}

\subsection{Paper outline.}
In Section~\ref{sec:maintaining}, we give first a short introduction to the data structures needed in our implementation, then we describe the implementations of the functions linkable (Section~\ref{sec:query}), insert and delete (Section~\ref{sec:updates}), and articulation-flip and separation-flip (Section~\ref{sec:flip}). Finally, Section~\ref{sec:flip-find} is dedicated to our implementation of $\oneflipquery$. 


\section{Maintaining a dynamic embedding}\label{sec:maintaining}
In this section we present a data structure
to maintain a dynamic embedding of a planar graph.  In the following,
unless otherwise stated, we will assume $G=(V,E)$ is a planar graph
with a given combinatorial embedding and that $\dual{G}=(F,\dual{E})$
is its dual.

Our primary goal is to be able to answer $\linkablequery(u,v)$, where $u$ and $v$ are vertices: Determine if an
edge between $u$ and $v$ can be added to $G$ without violating
planarity and without changing the embedding.  
If it can be inserted, return the list of pairs of \emph{corner}s (see
Definition~\ref{def:corner} below). Each corner-pair,  ($c_u$, $c_v$), 
describes a unique place where such an edge may be inserted.  If no such pair exists, return the empty list.

The data structure must allow efficient updates such as $\insertrm$, $\remove$,
$\vertexsplit$, $\vertexjoin$, and $\flip$.  We defer the exact
definitions of these operations to Section~\ref{sec:updates}.

As in most other dynamic graph algorithms we will be using a spanning
tree as the main data structure, and note:
\begin{observation}\label{obs:tree-cotree}[Tutte~\cite[p. 289]{tutte1984graph}]
  If $E_T\subseteq{}E$ induces a spanning tree $T$ in $G$, then
  $\dual{(E\setminus{}E_T)}$ induces a spanning tree $\cotree{T}$
  in~$\dual{G}$ called the \emph{co-tree} of $T$.
\end{observation}
\begin{lemma}\label{lem:cotree-cycle-has-result}
  If $u$ and $v$ are vertices, $T$ is a spanning tree, and $e$ is any
  edge on the $T$-path between $u$ and $v$, then any face containing
  both $u$ and $v$ lies
  on the fundamental cycle in $G^\ast$ induced by $\dual{e}$ in the co-tree $\cotree{T}$.
\end{lemma}
\begin{proof}
The fundamental cycle $C(e)$ induces a cut of the graph which separates $u$ and $v$. 

Let $f$ be a face which is incident to both $u$ and $v$. Then there will be two paths between $u$ and $v$ on the face $f$, choose one of them, $p$. Since $C(e)$ induces a cut of the graph which separates $u$ and $v$, some edge $e_C\in p$ must belong to the cut. But then, $f$ is a dual endpoint of $e_C^\ast$, and thus, $f$ belongs to $C(e)$.
\qed
\end{proof}

Thus, the main idea is to search a path in the co-tree.  This is
complicated by the fact (see Figure~\ref{fig:cotree-search-hard}) that the set of faces that are adjacent to $u$
and/or to $v$ need not be contiguous in $\cotree{T}$, so it is
possible for the cycle to change arbitrarily between faces that are
adjacent to any combination of neither, one, or both of $u$, $v$.
\begin{SCfigure}[50]
\centering
\includegraphics[width=0.45\linewidth]{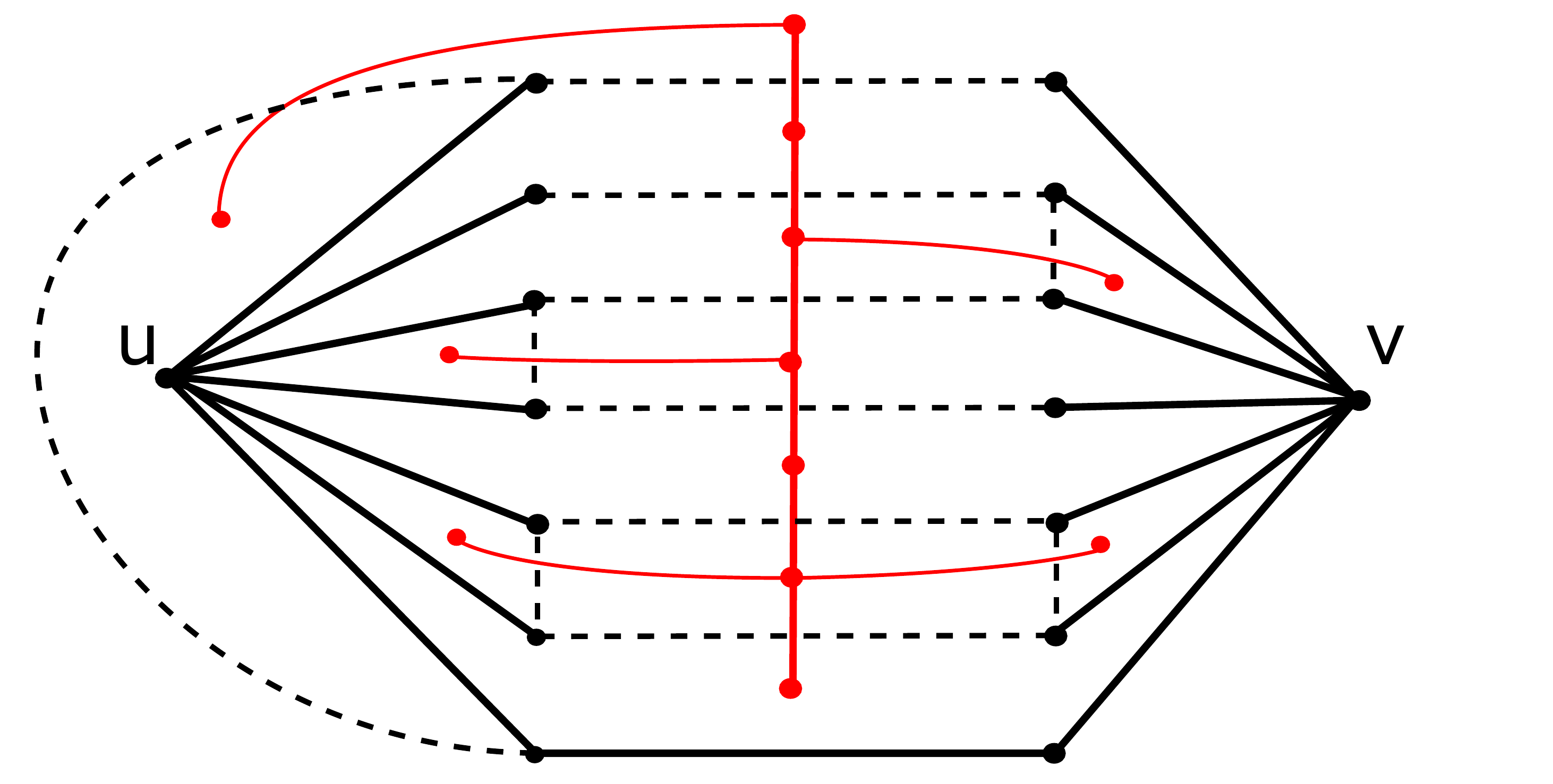}
\caption{The co-tree path may switch arbitrarily between faces
  that are adjacent to any combination of neither, one, or both of
  $u$, $v$.}
\label{fig:cotree-search-hard}
\end{SCfigure}

Our linkable query will consist of two phases. A marking phase, in which we ``activate'' (mark) all corners incident to each of the two vertices we want to link (see Section~\ref{sec:marking}), and a searching phase, in which we search for faces incident to ``active'' (marked) corners at both vertices (see Section~\ref{sec:query}). But first, we define corners.

\subsection{Corners and the extended Euler tour}\label{sec:corners}

The concept of a corner in an embedded graph turns out to be very
important for our data structure.  Intuitively it is simply a place in
the edge list of a vertex where you might insert a new edge, but as we
shall see it has many other uses.
\begin{definition}\label{def:corner}
  If $G$ is a non-trivial connected, combinatorially embedded graph, a
  \emph{corner} in $G$ is
  a $4$-tuple $(f,v,e_1,e_2)$ where $f$ is a face, $v$ is a
  vertex, and $e_1$, $e_2$ are edges (not necessarily distinct) that
  are consecutive in the edge lists of both $v$ and~$f$. 

  If $G = (\{v\}, \emptyset)$ and $G^\ast = (\{f\},\emptyset)$, there
  is only one corner, namely the tuple $(f,v,\nil,\nil)$.
\end{definition}
Note that faces and vertices appear symmetrically in the definition. Thus, there is a one-to-one correspondence between the corners of $G$ and the corners of $\dual{G}$.
This is important because it lets us work with corners in $G$ and
$\dual{G}$ interchangeably. Another symmetric structure is the \emph{Extended Euler Tour} (defined in Definition~\ref{def:EET} below). Recall that an Euler Tour of a tree is constructed by doubling all edges and finding an Eulerian Circuit of the resulting multigraph.
This is extended in the following the following way: 
\begin{definition} \label{def:EET}
   Given a spanning tree $T$ of a plane graph $G$, an \emph{oriented} Euler tour is one where consecutive tree edges in the tour are consecutive in the ordering around their shared vertex.
   
   Given the oriented Euler tour of the spanning tree $T$ of a non-trivial connected graph $G$, we may define the \emph{extended Euler tour} $\EET(T)$ by expanding between each consecutive pair of edges ($e_-,e_+$) with the list of corners and non-tree edges that come between $e_-$ and $e_+$ in the orientation around their shared vertex.
   If $G $ is a single-vertex graph, its extended Euler tour consist solely of the unique corner of $G$. 
\end{definition}

Thus, $\EET(T)$ is a cyclic arrangement 
that contains each edge in $G$ exactly twice and each corner 
in $G$ exactly once (see Figure~\ref{fig:euler_tour}).
Even more interesting is:
\begin{lemma}\label{lem:corner-euler-tour}
Given a spanning tree $T$ of $G$ with extended Euler tour $\EET(T)$, the corresponding tour $\EET(\cotree{T})$ in
$\dual{G}$ defines exactly the opposite cyclical arrangement of the
corresponding edges and corners in $\dual{G}$.
\end{lemma}
In order to prove the lemma, we use the following definition of an \emph{Euler cut}:
\begin{definition}
	Given any closed simple curve $C$ on the sphere $S^2$, then $S^2 \setminus C$ has two components, $A$ and $B$, each homeomorphic to the plane\footnote{Jordan curve Theorem.}. 
	Given an edge $e$ of the embedded connected graph $G$ with dual graph $G^\ast$, and a closed simple curve $C$, we say that $C$ intersects $e$ if $C$ intersects either $e$ of $G$, or $e^\ast$ of $G^\ast$. 
	Given a tree/cotree decomposition for $G$, we call $C$ an \emph{Euler cut} if the tree is entirely contained in $A$ and the cotree is entirely contained in $B$, and the curve only intersects each edge, $e$ or $e^\ast$, twice. 
\end{definition}

\begin{proof}[Proof of Lemma~\ref{lem:corner-euler-tour}] 
	Given any Euler cut $C$, note that the edges visited along the Euler tour of $G$ come in exactly the same order as they are intersected by $C$, if we choose the same ordering of the plane (clockwise or counter-clockwise). The placement of the corners is uniquely defined by the ordering of the edges.
	
	Such an Euler cut always exists, because both tree and cotree are connected, and because they share no common point. One can construct it by contracting one of the two trees to a point, and then considering an $\varepsilon$-circle $B_\varepsilon$ around that point, where $\varepsilon$ is small enough such that there is a minimal number of edge-intersections with that circle. The pre-image of $B_\varepsilon$ under the contraction will be an Euler cut. (See Figure \ref{fig:euler_cuts}.)
	
	Now, the Euler tour of $G$ is uniquely given by $C$. But we also know that $G$ is the dual of $G^\ast$, so $C$ is also an Euler cut for $G^\ast$, and thus, the Euler tour of $G^\ast$ is uniquely given by $C$. 
	
	Lastly, notice that in order for $C$ to determine the same ordering of the edges, we need to orient the two components oppositely, that is, one clockwise and one counter-clockwise.
	\qed
\end{proof}

\begin{figure}[H]
	\centering
	\includegraphics[width = 0.45\textwidth]{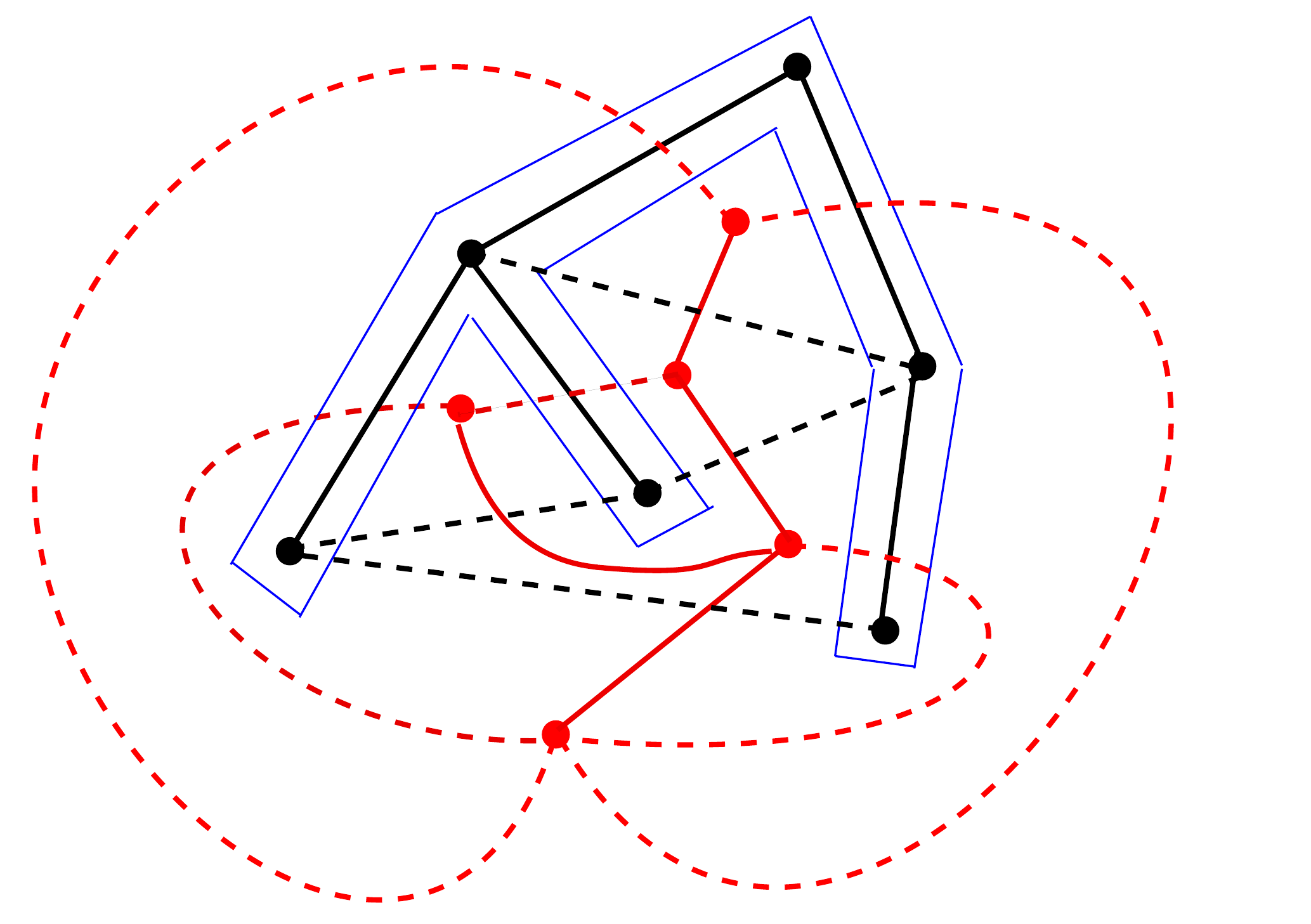}\hspace{0.1\textwidth}\includegraphics[width = 0.45\textwidth]{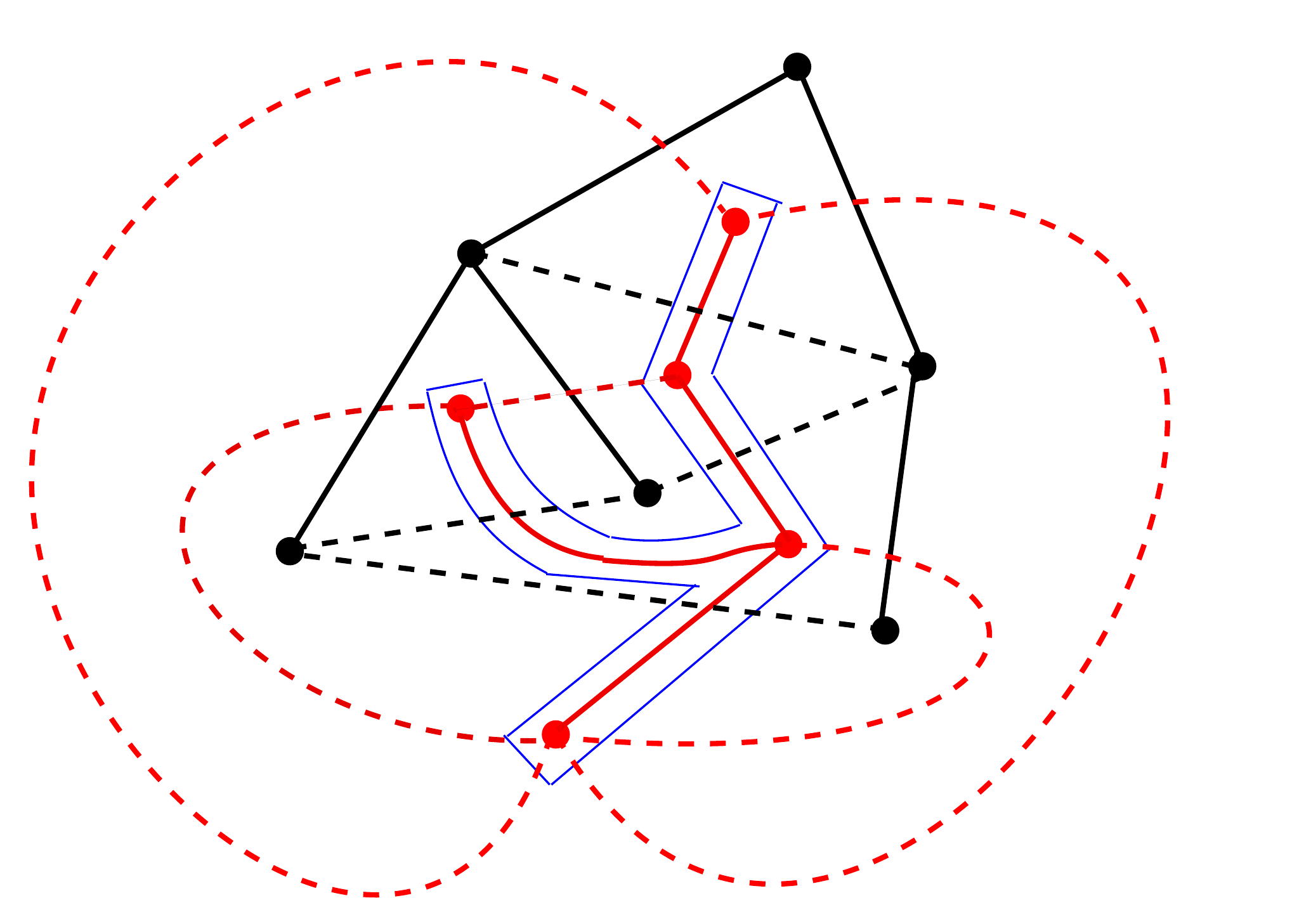}
	\caption{Two Euler cuts (blue) separating the primal tree from the dual tree. }
	\label{fig:euler_cuts}
\end{figure}
Thus, segments of the extended Euler tour translate directly between $T$
and $\cotree{T}$.  By \emph{segment}, we mean any contiguous sub-list
of the cycle.

\begin{SCfigure}[50]
\caption{This graph has extended Euler tour  \texttt{1a2b3c4d5b6e7d8f9e}$\ldots$\texttt{18h}, or, to write only the edges: \textbf{ab}cd\textbf{be}df\textbf{e}g\textbf{ah}gf\textbf{i}c\textbf{ih}. Edges not in the spanning tree are drawn with dotted lines.} 
\includegraphics[width=0.32\textwidth]{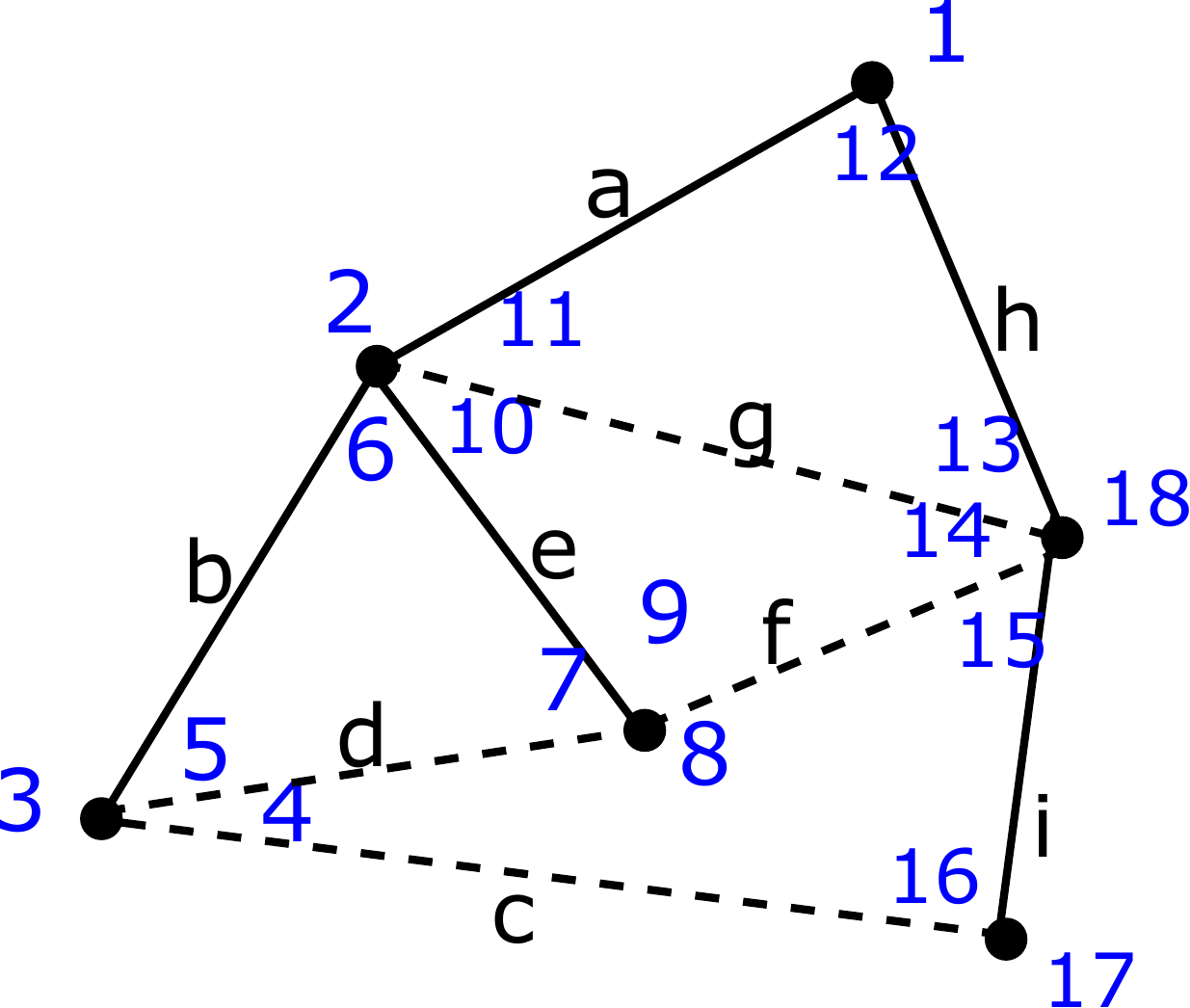}
\label{fig:euler_tour}
\end{SCfigure}

The high level algorithm for $\linkablequery(u,v)$, that is, to find a face in the embedding incident to both $u$ and $v$, (to be explained in detail later,) is now to build a structure consisting of an arbitrary
spanning tree $T$ and its co-tree $\cotree{T}$ such that we can

\begin{enumerate}[noitemsep,topsep=0pt,parsep=0pt,partopsep=0pt]
\item Find an edge $e$ on the $T$-path between $u$ and $v$.  This is
  easy.

\item Mark all corners incident to $u$ and $v$ in $T$.
  This is complicated by the existence of vertices of high degree, so
  a lazy marking scheme is needed.  However, it is easier
  than marking them in $\cotree{T}$ directly, since each vertex has a
  unique place in $T$ and no place in $\cotree{T}$.

\item Transfer those marks from $T$ to $\cotree{T}$ using
  Lemma~\ref{lem:corner-euler-tour}.  We can do this as long as
  the lazy marking scheme works in terms of segments of the extended
  Euler tour.

\item Search the fundamental cycle induced by $\dual{e}$ in $\cotree{T}$ for faces
  that are incident to a marked corner on both sides of the path.
\end{enumerate}

\subsection{Marking scheme}\label{sec:marking}

We need to be able to mark all corners incident to the two query
vertices $u$ and $v$, and we need to do it in a way that operates 
on segments of the extended Euler tour.  To this end we will 
maintain a modified version of a top tree over $T$ (see Alstrup et 
al.~\cite{DBLP:journals/talg/AlstrupHLT05} for an introduction to top trees), called the extended embedded top tree.

\subsubsection{Embedded top trees.}

Given a tree $T$, a top tree for $T$ is a binary tree. At the 
lowest level, its leaves are the edges of $T$. Its internal nodes, called \emph{clusters}, are sub-trees of $T$ with at most two boundary vertices. The set of boundary vertices of a cluster $C$ is denoted $\partial C$. At the highest level its root is a single cluster containing all of $T$. A non-boundary vertex of the subset $S\subset V$ may not be adjacent to a vertex of $V\setminus S$.
A cluster with two boundary vertices is called a \emph{path cluster}, and other clusters are called \emph{leaf clusters}.
Any internal node is formed by the merged union of its (at
most two) children. 
All operations on the top tree are implemented by
first breaking down part of the old tree from the top with
$\Oo(\log n)$ calls to a \emph{split} operation, end then building
the new tree with $\Oo(\log n)$ calls to a \emph{merge} operation.  
(This was proven by Alstrup et al. in \cite{DBLP:journals/talg/AlstrupHLT05}.)

The operation \emph{expose} on the top tree takes one or two vertices and makes them boundary of the root cluster. 
%
We may even expose any constant number of vertices, 
and then, instead of having one top tree with the cluster containing all of $T$ in the root, we obtain a constant number of top trees with those vertices  
exposed.


\begin{definition}
  An \emph{embedded top tree} is a top tree over an embedded tree $T$
  with the additional property that whenever it is merging $\{C_i\}$
  to form $C$, for each pair $C_i,C_{i+1}$ there exist edges $e_i\in
  C_i$ and $e_{i+1}\in C_{i+1}$ that are adjacent in the cyclic order
  around the common boundary vertex.
\end{definition}

\begin{definition}
  Given an embedded graph $G$ and a spanning tree $T$ of $G$, an
  \emph{extended embedded top tree} of $T$ is an embedded top tree over
  the embedded tree $T'$ defined as follows:
  \begin{itemize}
  \item The vertices of $T'$ are the vertices of $G$, plus one extra
    vertex for each corner of $G$ and two extra vertices for each
    non-tree edge (one for each end vertex of the edge).
  \item The edges of $T'$ are the edges of $T$, together with edges
    connecting each vertex $v$ of $G$ with the vertices representing
    the corners and non-tree edges incident to $v$.
  \item The cyclic order around each vertex is the order inherited 
    from the embedding of $G$.
  \end{itemize}
\end{definition}

\begin{observation}\label{obs:segments}
  In any extended embedded top tree over a tree $T$, path clusters
  consist of the edges and corners from two segments of $\EET(T)$, and
  leaf clusters consist of the edges and corners from one segment of
  $\EET(T)$.
\end{observation}

\begin{lemma}
  We can maintain an extended embedded top tree over $T$ using
  $\Oo(\log n)$ calls to merge and split per update.
\end{lemma}
\begin{proof}
  First note that we can easily maintain $T'$.  Now construct the
  ternary tree $\widehat{T'}$ by substituting each vertex by its chain
  (see Italiano et al.~\cite{DBLP:conf/esa/ItalianoPR93}).  Consider
  the top tree for $T'$ made via reduction to the topology tree for
  $\widehat{T'}$ (see Alstrup et
  al.~\cite{DBLP:journals/talg/AlstrupHLT05}).  Since the chain is
  defined to have the vertices in the same order as the embedding,
  this top tree is naturally an embedded top tree. \qed
\end{proof}

For the rest of this paper, we will assume all top trees are extended
embedded top trees, without any further mention.

\begin{observation}\label{obs:exposecorner}
Since corners in $T$ are represented by vertices in $T'$, we automatically get the ability to \emph{expose} corners, giving complete control over which $\EET(T)$ segment is available for information or modification.
\end{observation}

\subsubsection{Slim-path top trees and four-way merges.}

We can tweak the top tree further such that the following property, called \emph{the slim path invariant}, is maintained: 
\begin{itemize}
\item for any path-cluster, all edges in the cluster that are incident to a boundary vertex belong to the cluster path. In other words, for each boundary vertex $b$, there is at most one edge in the cluster that is incident to $b$.
\end{itemize}
This can be done by allowing the following merge:
Given path-clusters $C,C'$ with boundary vertices $\partial C = \{u,v\}$ and $\partial C' = \{v, w\}$, and given two leaf-clusters with boundary vertex $v$, allow the merge of those four, to obtain one path cluster with boundary vertices $\{u,w\}$. For lack of a better word, we call this a \emph{four-way merge}.

\begin{figure}[H]
	\centering
	\includegraphics[width=0.3\linewidth]{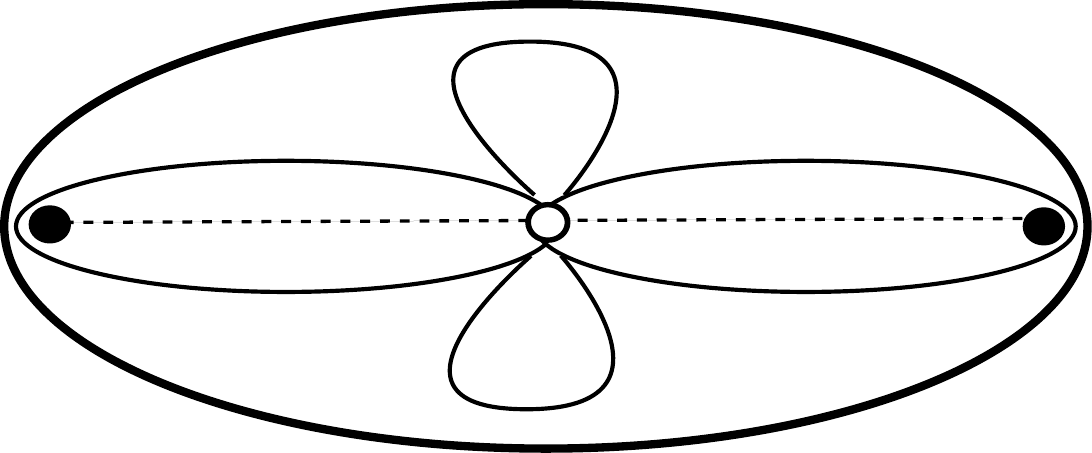}
	\caption{A four-way merge of $2$ paths and $2$ leaves.}
	\label{fig:belatedmerge}
\end{figure}


\noindent We allow the following merges:
\begin{itemize}
	\item Two leaves merge to one,
	\item A path cluster with boundary vertices $u,v$, and a leaf cluster with boundary vertex $v$, merge to form a leaf with boundary vertex $u$, 
	\item Two path clusters with boundary vertices $u, v$, and $v, w$, respectively, and up to two leaves with boundary vertex $v$ make a four-way merge to form a path cluster with boundary vertices $u,v$ (see Figure~\ref{fig:belatedmerge}).
\end{itemize}

Since the slim-path invariant has no restriction on leaves, clearly the first two do not break the invariant. The four-way merge does not break the invariant: If, before the merge, there was only one edge incident to $u$ and only one incident to $w$, this is still maintained after the merge.

\begin{definition}
	A \emph{slim-path top tree} is a top tree which maintains the slim-path invariant by allowing four-way merges.
\end{definition}

\begin{lemma}\label{lem:slimpath}
	We can maintain a slim-path top tree over a dynamic tree with height $\Oo(\log n)$, using $\Oo(\log n)$ calls to $\clustermerge$ and $\clustersplit$ per update.
\end{lemma}

\begin{proof} 
	Proof by reduction to top trees (see figure \ref{fig:fourway}). For each path cluster in the ordinary top tree, keep track of a slim path cluster, and up to four internal leaves; two for each boundary vertex, one to each side of the cluster path. In the slim-path top tree, this would correspond to four different clusters, but this doesn't harm the asymptotic height of the tree. When merging two path clusters in the ordinary top tree, with paths $u\tto v$ and $v\tto w$, this simply correspond to two levels of merges in the slim-path top tree: First, the incident leaf clusters are merged, pairwise, and then, we perform a four-way merge. Other merges in the ordinary top tree correspond exactly to merges in the slim-path top tree. Thus, each merge in the ordinary top tree corresponds to at most $2$ levels of merges in the slim-path top tree, and the maintained height is $\Oo(\log n)$.
	\qed
\end{proof}

\begin{figure}[H]
	\centering
	\includegraphics[width=0.9\linewidth]{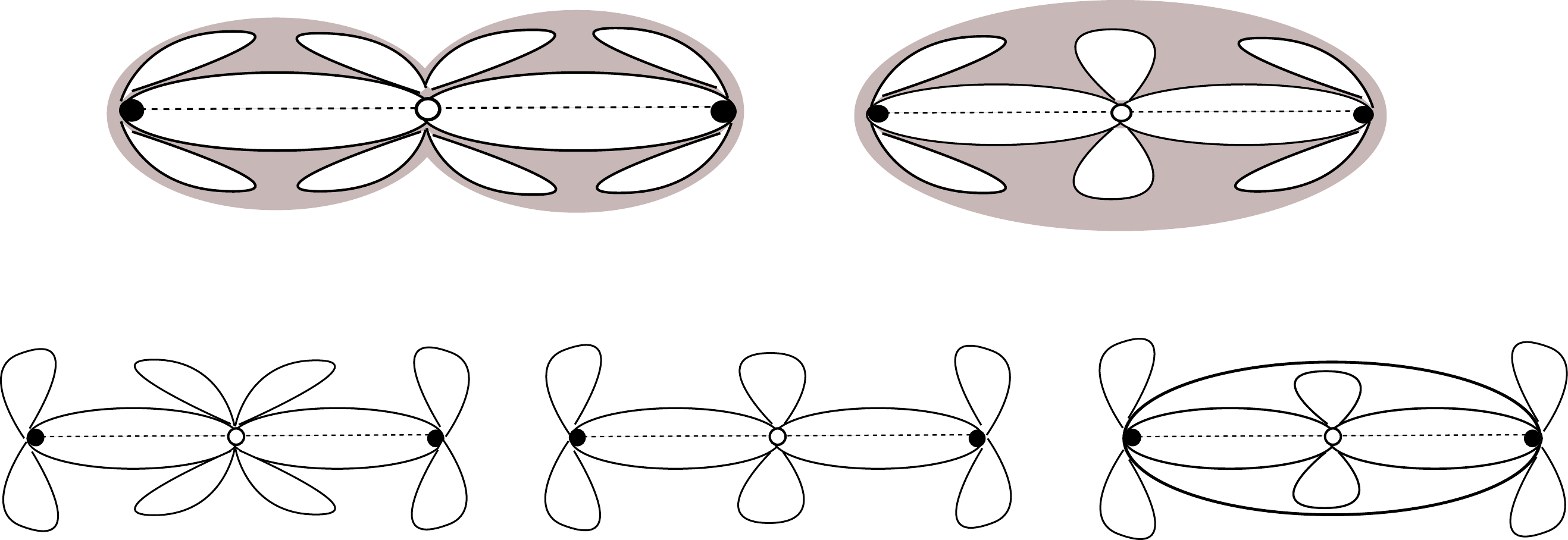}
	\caption{One merge in ordinary top trees corresponds to two merges in slim-path top trees.}
	\label{fig:fourway}
\end{figure}

Note that four-way merges and the slim-path invariant are not necessary. One could simply maintain, within each path cluster, information about the internal leaf clusters incident to each boundary vertex; one to each side of the cluster path. We present four-way merges because they provide some intuition about this particular form of usage of top trees.	

\begin{figure}[H]
	\centering
	\includegraphics[width=0.5\linewidth]{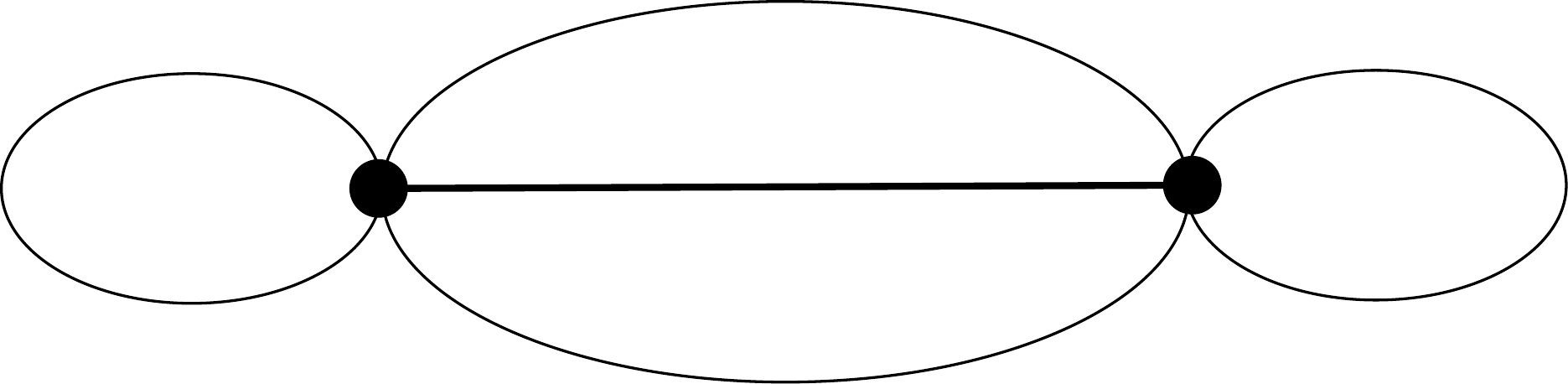}
	\caption{In a slim-path top tree, expose($u,v$) exposes a path and two leaves.}
	\label{fig:exposeslimpath}
\end{figure}

We also note that the expose($u,v$)-operation is different in a slim-path top tree, whenever $u$ or $v$ is not a leaf in the tree (see Figure~\ref{fig:exposeslimpath}.) Instead of returning a top tree which has at its top a path cluster with path $u\tto v$, the operation returns up to three top trees: one  with at its root a path cluster with path $u\tto u$, and, up to two top trees with leaf clusters at their roots, one for each exposed non-leaf vertex.


\begin{figure}[H]
\centering
\includegraphics[width=0.7\linewidth]{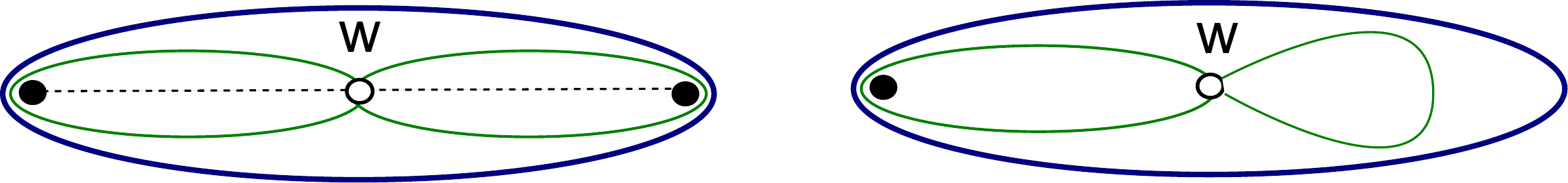}
\caption{The vertex $w$ is a boundary vertex of the green clusters, but not of their blue parent clusters.}
\label{fig:clustermerges}
\end{figure}

\begin{lemma}\label{lem:bnode-discard-segments}
  Whenever a merge of two clusters in the slim-path top tree causes a vertex $w$
  to stop being a boundary vertex (see Figure~\ref{fig:clustermerges}), all corners incident to $w$ are
  contained in one or two $\EET(T)$ segments. 
  These segments will be sub-segments of the (one or two) segments corresponding to the parent cluster $C$ (blue in Figure~\ref{fig:clustermerges}), and will not contain any corners incident to the (one or two) boundary vertices of $C$.
\end{lemma}
\begin{proof}
	There are two statements to prove. Firstly, all corners incident to $w$ are contained in the one or two $\EET(T)$ segments described in Observation~\ref{obs:segments}. Secondly, it follows from the slim-path invariant that the delimiting corners for any boundary vertex $v\in\partial C$ are exactly those incident to the unique tree-edge in $C$ incident to $v$.
\qed
\end{proof}

The following term is defined for both top trees and slim-path top trees and is very convenient:
\begin{definition}\label{def:rootpath} 
	Given a vertex $v$, let $C_v$ denote the deepest cluster in the top tree where $v$ is not a boundary vertex. The \emph{root path} of $v$ in the top tree is the path from $C_v$ to the root in the top tree. For a corner $c$ incident to a vertex $v$, the \emph{root path} of $c$ is the root path of $v$.
\end{definition}

\subsubsection{Deactivation counts}

Now suppose we associate a (lazy) \emph{deactivation count} with each
corner. The deactivation count is set to $0$ before we start building the top tree. Define the
$\clustermerge$ operation on the top tree such that whenever a merge
discards a boundary vertex we \emph{deactivate} all corners on the at
most two segments of $\EET(T)$ mentioned in Lemma~\ref{lem:bnode-discard-segments} by increasing that count
(and define the $\clustersplit$ operation on the top tree to
reactivate them as necessary).  When the top tree is complete, the
corners that are still \emph{active} (have deactivation count $0$) are
exactly those incident to the boundary vertices of the root of the
top tree.  These boundary vertices are controlled by the $\expose$
operation on the top tree and changing the boundary vertices require
only $\Oo(\log n)$ merges and splits, so we have now argued the
following:
\begin{lemma}
  We can mark/unmark all corners incident to vertices $u$ and $v$ by
  increasing and decreasing the deactivation counts on 
  $\Oo(\log n)$ segments of the extended Euler tour.
\end{lemma}

What we really want, is to be able to search for the marked 
corners in $\cotree{T}$, so instead of storing the counts (even 
lazily) in the top tree over $T$, we will store them in a top 
tree over $\cotree{T}$.
Again, each cluster in this top tree covers one or two segments 
of the extended Euler tour. 
For each segment $S$ we keep track 
of 
a number $\Delta(S)$ such that for every corner, $c$, 
the deactivation count for $c$ equals the sum of $\Delta(S)$ over all segments $S$ containing $c$ in the root path of $c$ .
\begin{align*}
\text{deactivation count}(c) = \sum_{\substack{\text{cluster } C\in \text{root path}(c),\\\text{segment } S\in C \,\mid \, c\in S}} \Delta(S)
\end{align*}
In addition, for each boundary vertex (face) $f\in\partial C$, we keep track the value $c_{\min}(C,f)$, defined as the minimum deactivation count over all corners in the cluster $C$ that are incident to $f$.

To update the
deactivation counts of an arbitrary segment $S$, all we need to do is
modify the $\Oo(\log n)$ clusters that are affected, which can be
done in $\Oo(\log n)$ time, leading to
\begin{lemma}\label{lem:7}
  We can maintain a top tree over $\cotree{T}$ that has
  $c_{\min}(C,f)$ for each boundary vertex $f\in\partial C$ where $C$ is a cluster in $\cotree{T}$,
  in $\Oo(\log^2 n)$ time per change in the set of (at most two) marked vertices.
\end{lemma}

\begin{proof} To mark a given set of vertices, expose them in the top tree for $T$. For each merge of clusters $\{C_i\}$ into a new cluster $C$ in the top tree over $T$, where $\partial C\neq\bigcup\partial C_i$, consider the at most two segments of $\EET(T)$ associated with $C$. For each such segment, $S$, expose the delimiting corners of $S$ in the top tree over $\cotree{T}$ (see Observation~\ref{obs:exposecorner}), and increase $\Delta(R)$ for the relevant segment $R$ of the root cluster (in the dual top tree). 
	
For each merge in the dual top tree, where $\{C_i\}$ merge to form $C$, and for each boundary vertex $f\in \partial C$, set $c_{\min}(C,f)=\min\{c_{\min}(C_i,f) \,\mid\, f\in \partial C_i\}$.

For each split of a cluster $C$ into clusters $\{C_i\}$, and each segment $S$ in $C$ corresponding to segments $\{S_i\}$ with $S_i$ in $C_i$, propagate
$\Delta(S)$ down by setting $\Delta(S_i) = \Delta(S_i)+\Delta(S)$.

Since expose in the primal top tree can be implemented with $\Oo(\log n)$ calls to $\clustermerge$ and $\clustersplit$, this will yield at most $\Oo(\log n)$ calls to expose in the dual top tree, and thus, in total, the time for marking a set of vertices is $\Oo(\log ^2 n)$. \qed
\end{proof}

\begin{observation}\label{obs:incidence}
This is enough for, given a face $f$ and a vertex $u$, checking whether $f$ is incident to $u$ in $\Oo(\log ^2 n)$ time.
\end{observation}
 
\begin{proof}
  Expose $u$ in the primal top tree. This takes $\Oo (\log ^2 n)$ time. Expose $f$ in the dual top tree. This takes $\Oo(\log n)$ time. 
  Then, for the root cluster $R$ which has boundary vertex $f$, $c_{\min}(R,f)=0$ if an only if $f$ and $u$ are incident.
  \qed 
\end{proof}

\subsection{Linkable query}\label{sec:query}

Unfortunately, the $c_{\min}$ and $\Delta$ values discussed in
Section~\ref{sec:marking} are not quite enough to let us find the
corners we are looking for.  We can use them to ask what marked corners
a given face is incident to, but we do not have enough to find
\emph{pairs} of marked corners on opposite sides of the same face on
the co-tree path.

As noted in Lemma~\ref{lem:cotree-cycle-has-result}, for two given vertices $u$ and $v$, there exists a path in the dual tree containing all candidates for a common face. And this path is easily found! Since the dual of a primal tree edge induces a fundamental cycle that separates $u$ and $v$, we may use the path between the dual endpoints $f,g$ of any edge on the primal tree path between $u$ and $v$. Furthermore, once we expose $(f,g)$ in the dual tree, if $f\neq g$, the root will have two EET-segments: the minimum deactivation count of one EET-segment is $0$ if and only if any non-endpoint faces are incident to $v$, the other is $0$ if and only if any are incident to $u$. Checking the endpoint faces can be done (cf. Observation~\ref{obs:incidence}), but to find non-endpoint faces we need more structure.

To just output \emph{one} common face, our solution is for each path cluster in the top tree over the co-tree
to keep track of \emph{a single} internal face $f_{\min}$ on the cluster
path that is incident to minimally deactivated corners on either side
of the cluster path if such a face exists.
For that purpose, we define for any EET-segment $S$ the value $c_{\min}(S)$ to represent the minimal deactivation count of a corner in $S$.

\begin{figure}[H]
	\centering
	\includegraphics[width=0.7\linewidth]{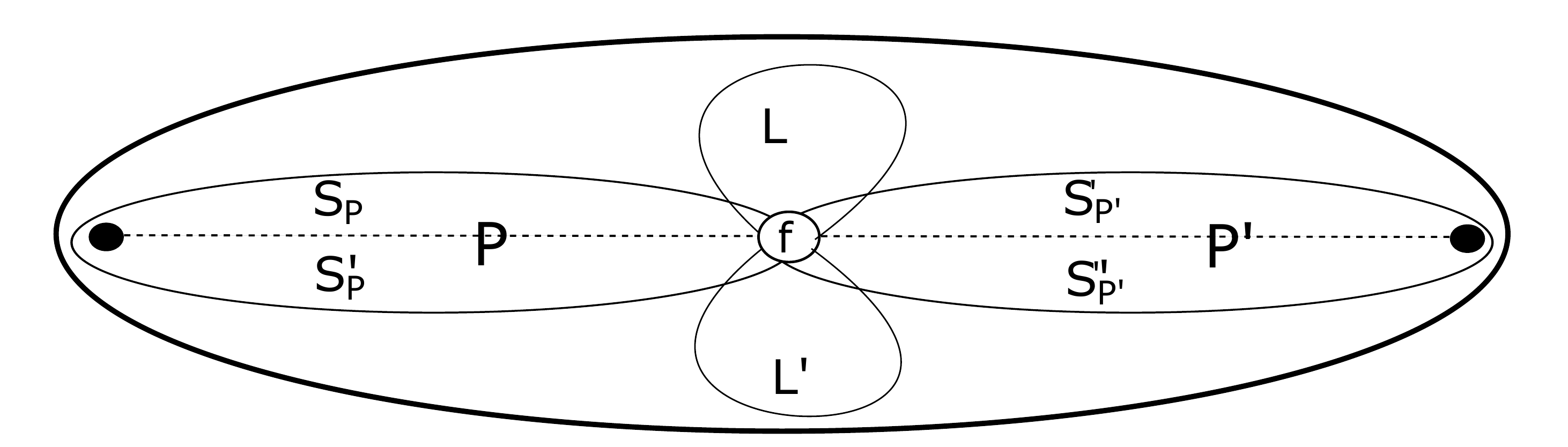}
	\caption{After a four-way cluster merge of $L,L',P,P'$ to $C$, the face $f_{\min}(C)$, if it exists, is found among $f_{\min}(P),f_{\min}(P'),$ and $f$.}
	\label{fig:pathclustermerges}
\end{figure}
\begin{lemma}
  We can maintain a top tree over $\cotree{T}$ that has
  $c_{\min}$ and $\Delta$ values for each EET-segment in each
  cluster and $f_{\min}$ values for each path cluster in
  $\Oo(\log^2 n)$ time per change to the set of (at most two) marked vertices.
\end{lemma}
\begin{proof}
  We maintain a slim-path top tree for the dual tree (see Lemma~\ref{lem:slimpath}).
	
  Merge and split of leaf-clusters are handled as in the proof of Lemma~\ref{lem:7}. When two leaf clusters $L,L'$ with segments $S_L, S_{L'}$ merge, $c_{\min}$ for the resulting segment is set to be $\min\{c_{\min}(S_L),c_{\min}(S_{L'})\}$.

  For a $4$-way $\clustermerge$ (see Figure~\ref{fig:pathclustermerges}) of two path clusters $P,P^{\prime}$ and two leaf clusters $L,L'$ with common boundary vertex (face) $f$ into a cluster $C$, we need to calculate $c_{\min}(S_C)$ and $c_{\min}(S'_C)$ for its two EET-segments, and the value of $f_{\min}(C)$.
  The EET-segment $S_C$ is the concatenation of three EET-segments, $S_P$ from $P$, $S_{p^{\prime}}$ from $P'$, and $S_L$ from $L$. We may then set $c_{\min}(S_C) = \min\{c_{\min}(S_p), c_{\min}(S_{p^{\prime}}), c_{\min}(S_L)\}$.
  
  Similarly, $c_{\min}(S'_C) = \min\{c_{\min}(S'_p), c_{\min}(S'_{p^{\prime}}), c_{\min}(S_{L'})\}$. 
  
  We may now assume we have calculated the values $c_{\min}(S_C)$ and $c_{\min}(S'_C)$. 
  To determine $f_{\min}(C)$, one only has to check the three values: $f_{\min}(P), f_{\min}(P^{\prime})$, and $f$.
  If $c_{\min}(S_C) = c_{\min}(S_P)$, and $c_{\min}(S'_C) = c_{\min}(S'_P)$, then we may set $f_{\min}(C) = f_{\min}(P)$.
  Otherwise, if $c_{\min}(S_C) = c_{\min}(S_{P'})$, and $c_{\min}(S'_C) = c_{\min}(S'_{P'})$, then we may set $f_{\min}(C) = f_{\min}(P')$. Otherwise, if $c_{\min}(S_C)=c_{\min}(L)$ and $c_{\min}(S'_C)=c_{\min}(L')$, we may set $f_{\min}(C) = f$. Otherwise, finally, we have to set $f_{\min}(C) = \nil $.
%
  \qed 
\end{proof}

\begin{lemma}
  We can support each $\linkablequery(u,v)$ in $\Oo(\log^2 n)$ time per operation.
\end{lemma}
\begin{proof}
  If $u$ and $v$ are not in the same connected component we pick any corners $c_u$
  and $c_v$ adjacent to $u$ and $v$ and return them.  Otherwise, we use
  $\expose(u,v)$ on the top tree over $T$ to activate all corners adjacent
  to $u$ and $v$ and to find an edge $e$ on the $T$-path from $u$ to
  $v$ (e.g. the first edge on the path). Let $g,h$ be the endpoints of $\dual{e}$, and call
  $\expose(g,h)$ on the top tree over $\cotree{T}$. By Lemma~\ref{lem:cotree-cycle-has-result}, a common face lies on the cotree path from $g$ to $h$. Let $f$ be the
  $f_{\min}$ value of the resulting root.  We can now test each
  of $f,g,h$ using the $c_{\min}$ values to find the desired
  corners if they exist. \qed 
\end{proof}
\begin{lemma}
  If there are more valid answers to $\linkablequery(u,v)$ we can find $k$ of 
  them in $\Oo(\log^2 n + k)$ time. 
\end{lemma}
\begin{proof}
   For each leaf cluster and for each side of each path cluster we can
   maintain the list of minimally deactivated corners adjacent to each
   boundary vertex.  Then, instead of maintaining a single face
   $f_{\min}$ for each path cluster, we can maintain a linked list
   of all relevant faces in the same time. And for each side of each
   face in the list we can point to a list of minimally deactivated
   corners that are adjacent to that side. For leaf-clusters, we point to a linked list of minimally deactivated corners incident to the boundary vertex.
   Upon the merge of clusters, face-lists and corner-lists may be linked together, and the point of concatenation is stored in the resulting merged cluster in case of a future split. Note that each face occurs in exactly one face-list.

   As before, to perform $\linkablequery(u,v)$, expose $u,v$ in the primal tree. Let $e_0$ be an edge on the tree-path between $u$ and $v$, and expose the endpoints of $e_0^{\ast}$ in the dual top tree. Now, the maintained face-list in the root of the dual top tree contains all faces incident to $u,v$, except maybe the endpoints of $e_0^{\ast}$, which can be handled separately, as before.
   The total time is therefore $\Oo(\log^2 n)$ for the
   necessary expose operations, and then $\Oo(1)$ for each reply.
   \qed
\end{proof}

\begin{observation}\label{obs:dualstruct}
If we separately maintain a version of this data structure for the dual graph, then for faces $f,g$, $\linkablequery(f,g)$ in that structure lets us find vertices that are incident to both $f$ and $g$.
\end{observation}

\subsection{Updates}\label{sec:updates}

In addition to the query, our data structure supports the following
set of update operations:
\begin{itemize}[noitemsep,topsep=0pt,parsep=0pt,partopsep=0pt]
\item $\insertrm(c_u,c_v)$ where $c_u$ and $c_v$ are corners that are
  either in different connected components, or incident to the same face.
  Adds a new edge to the graph, inserting it between the edges of
  $c_u$ at one end and between the edges of $c_v$ at the
  other. Returns the new edge.

\item $\remove(e)$.  Removes the edge $e$ from the graph.  Returns 
  the two corners that could be used to insert the edge again.

\item $\vertexjoin(c_u,c_v)$ 
  where $c_u$ and $c_v$ are corners that
  are either in separate components of the graph or in the same face.
  Combines the vertices $u$ and $v$ into a single new vertex $w$ and
  returns the two new corners $c_w$ and $c_w'$ that may be used to
  split it again using $\vertexsplit(c_w,c_w')$.

\item $\vertexsplit(c_w,c_w^\prime)$ where $c_w$ and $c_w^\prime$ are corners
  sharing a vertex $w$.  Splits the vertex into two new vertices $u$
  and $v$ and returns corners $c_u$ and $c_w$ that might be used to
  join them again using $\vertexjoin(c_u,c_v)$.
%
\end{itemize}

When calling $\remove(e)$ on a non-bridge tree-edge $e$, we need to search for a replacement edge. Luckily, $e^{\ast}$ induces a cycle in the dual tree, and any other edge on that cycle is a candidate for a replacement edge. 
If we like, we can augment the dual top tree so we can find
the minimal-weight replacement edge, simply let each path cluster remember the cheapest edge on the tree-path, and expose the endpoints of $e^{\ast}$. 
If we want to keep $T$ as a minimum spanning tree, we also need to
check at each $\insertrm$ and $\vertexjoin$ that we remove the
maximum-weight edge on the induced cycle from the spanning tree.

In general, when we need to update both the top trees over $T$ and
$\cotree{T}$ we must be careful that we first do the $\clustersplit$s
needed in the top tree over $T$ to make each unchanged sub-tree 
 a
(partial) top tree by itself, then update the top tree over
$\cotree{T}$ and finally do the remaining $\clustersplit$s and
$\clustermerge$s to rebuild the top tree over $T$.  This is necessary
because the $\clustermerge$ and $\clustersplit$ we use for $T$ depend
on $T$ and $\cotree{T}$ having related extended Euler tours.

\begin{wrapfigure}[9]{r}{0.42\textwidth}
\vspace{.5em}
\includegraphics[width=1.0\linewidth]{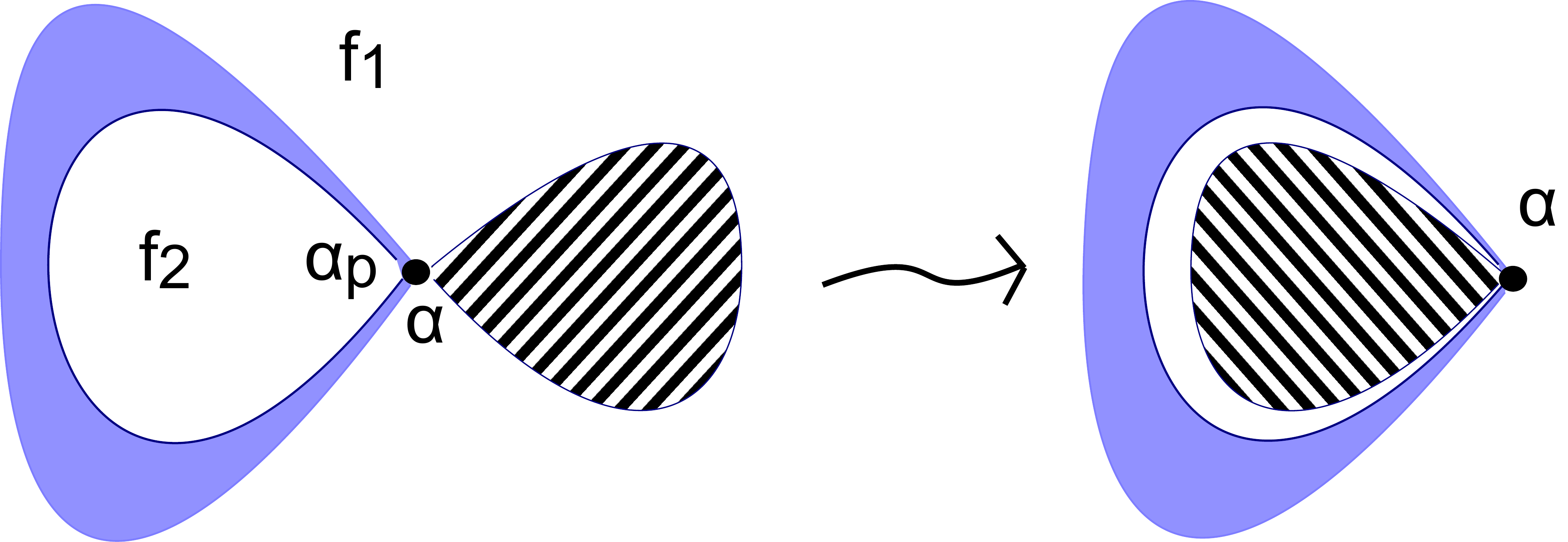}
\caption{An articulation flip at the vertex $\alpha$.}
\label{fig:articulation_flip}
\end{wrapfigure}
Any change to the graph, especially to the spanning tree, implies a change to the extended Euler tour. Furthermore, any deletion or insertion of an edge implies a merge or split in the dual tree. E.g. if an edge is inserted across a face, that face is split in two. As a more complex example, if the non-bridge tree-edge $e=(u,v)$ is deleted, the replacement edge is removed from the dual tree, and the endpoints of $e^{\ast}$ are merged.
%
\subsubsection{Cut and join of faces.}

Given a face $f$ and two corners incident to the face, say, $c_1$ and $c_2$, we may cut the face in those corners, producing new faces $f$ and $f^\prime$. We go about this by exposing $c_1,c_2$ in the dual top tree.
We then have two chains for $f$, which we rename as $f$ and $f^\prime$. 


Similarly, faces may be joined. Given a face $f$ and a face $g$, a corner $c_f$ incident to $f$ and a corner $c_g$ incident to $g$, we may join the two faces to a new face $h$, with two new corners $c_1$ and $c_2$, such that all edges incident to $h$ between $c_1$ and $c_2$ are exactly those from $f$, in the same order, and those between $c_2$ and $c_1$ are those from $g$. We do this by exposing $f$ and $c_f$ in one top tree, and $g$ and $c_g$ in the other top tree, and then by the link and join operation.

\subsubsection{Inserting an edge}

We now show how to perform $\operatorname{insert}(c_u,c_v)$, when compatible with the embedding.

To insert an edge between two components, we must add it to the primal tree. Given a vertex $v$ with an incident corner $c_v$ to the face $g$, and a vertex $u$ with an incident corner $c_u$ to the face $f$, we may  insert an edge $a$ between $u$ and $v$.

In each primal top tree, we expose $v,c_v$ and $u,c_u$, respectively. In the dual top tree, we expose $g,c_v$ and $f,c_u$, respectively. We then link $u$ and $v$, and update the EET correspondingly: New corners $c_1,\ldots ,c_4$ are formed, such that the new edge $a=(u,v)$ appears as the successor of $c_4$ and the predecessor of $c_1$, and the successor of $c_2$ and the predecessor of $c_3$ (see figure \ref{fig:connect-components}). That is, we may view the top tree in its current form as the top tree with an exposed path $(u,v)$ which is trivial (consisting of one edge) and two leaves, one at each end.

In the dual tree, the faces $f$ and $g$ are simply joined at the incident corners $c_u$ and $c_v$, respectively.

\begin{figure}[H]
	\centering
	\includegraphics[width=1.0\linewidth]{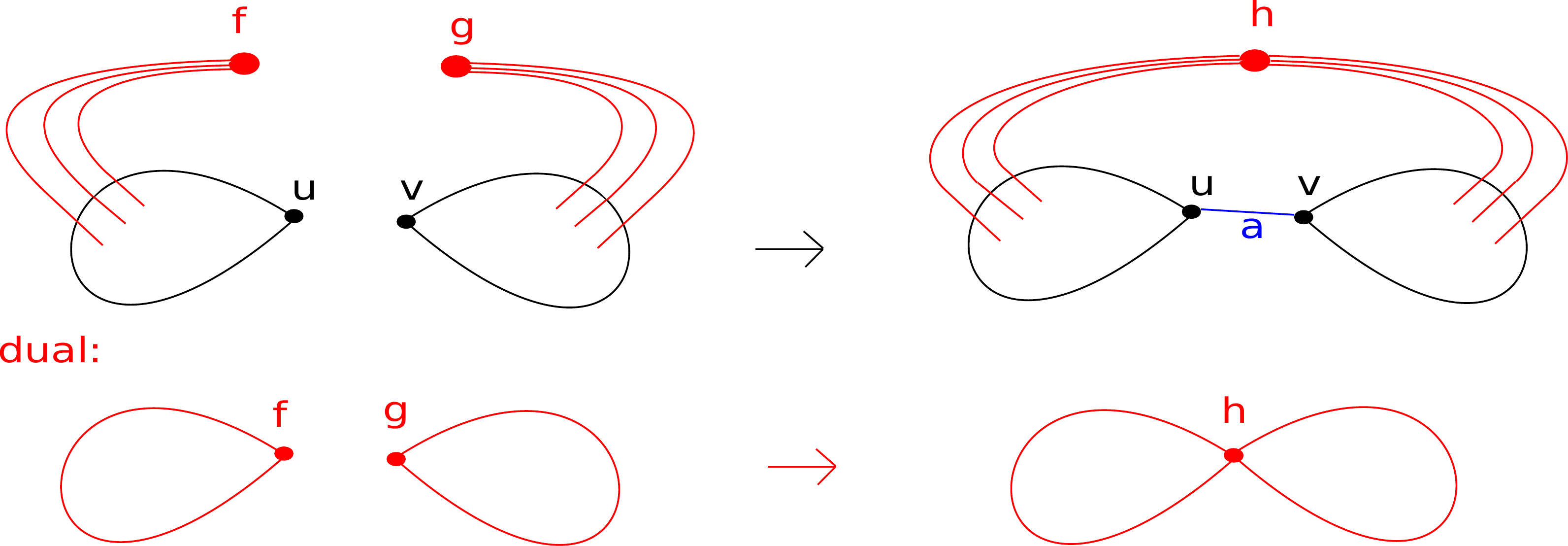}
	\caption{The component of $u$ is joined with the component of $v$ by the edge $a$, and new corners are formed.}
	\label{fig:connect-components}
\end{figure}

To insert an edge inside a component, the two vertices must belong to the same face. Let $u$ and $v$ be vertices with corners $c_u$ and $c_v$ incident to a face $f$. We can then expose the corners $c_u$ and $c_v$, which gives us a path cluster in the primal top tree (and some leaves, because we use slim-path top trees), and in the dual tree, exposing those corners gives us two leaves, both incident to $f$.

We now create new corners $c_1,\ldots,c_4$, and update the EET to list the new edge, $a=(u,v)$, as the successor of $c_1$ and predecessor of $c_2$, replacing $c_v$, and as the successor of $c_3$ and predecessor of $c_4$, replacing $c_u$ (see figure \ref{fig:connect-inside-face}). In the dual tree, the face $f$ is split in the corners $c_u$ and $c_v$, creating new faces $g$ and $h$, and a new edge $(g,h)$ is added to the dual tree. 

\begin{figure}[H]
	\centering
	\includegraphics[width=1.0\linewidth]{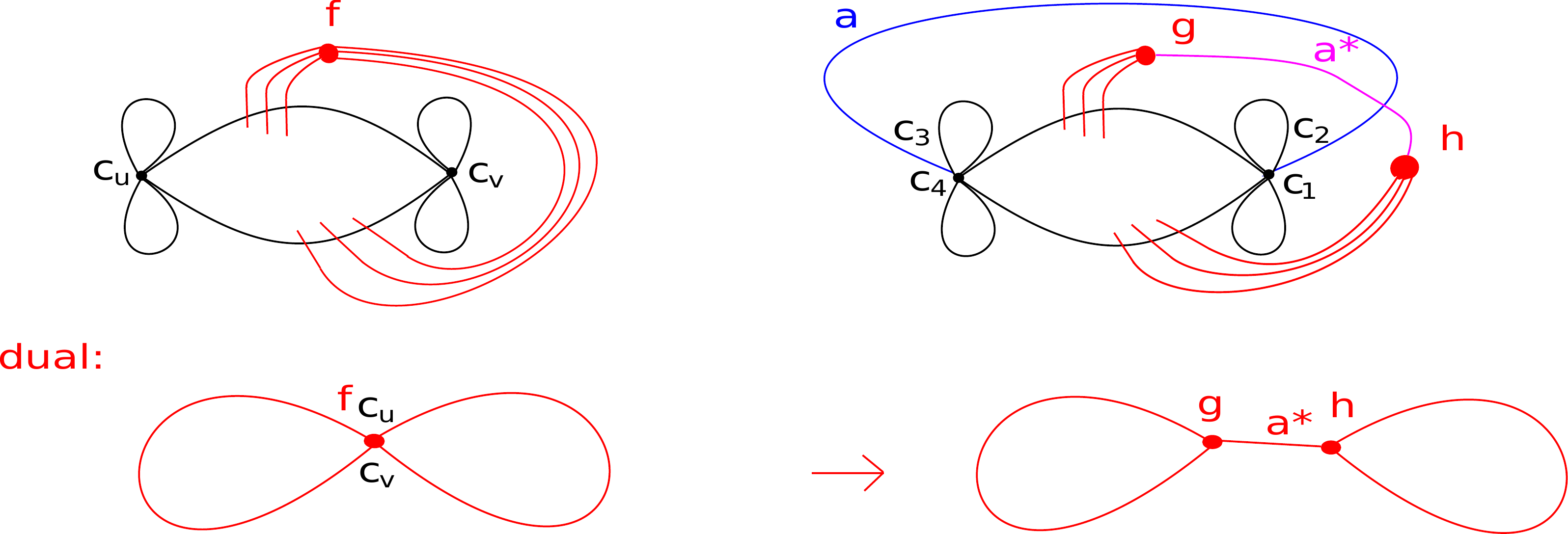}
	\caption{The corners $c_u$ and $c_v$ are exposed, and an edge $a=(u,v)$ (blue) is inserted in those corners. The dual of $a$, namely $a^\ast = (f,g)$ (pink), is inserted in the dual tree to connect the two new faces, that appear when the face $f$ is cut in two by the insertion of the edge $a$.}
	\label{fig:connect-inside-face}
\end{figure}

\subsubsection{Deleting an edge}

We show now how to perform $\operatorname{delete}(e)$.

There are three cases, depending on whether the edge to be deleted is a bridge, a non-tree edge, or a replacement edge exists.

\begin{figure}[H]
	\centering
	\includegraphics[width=1.0\linewidth]{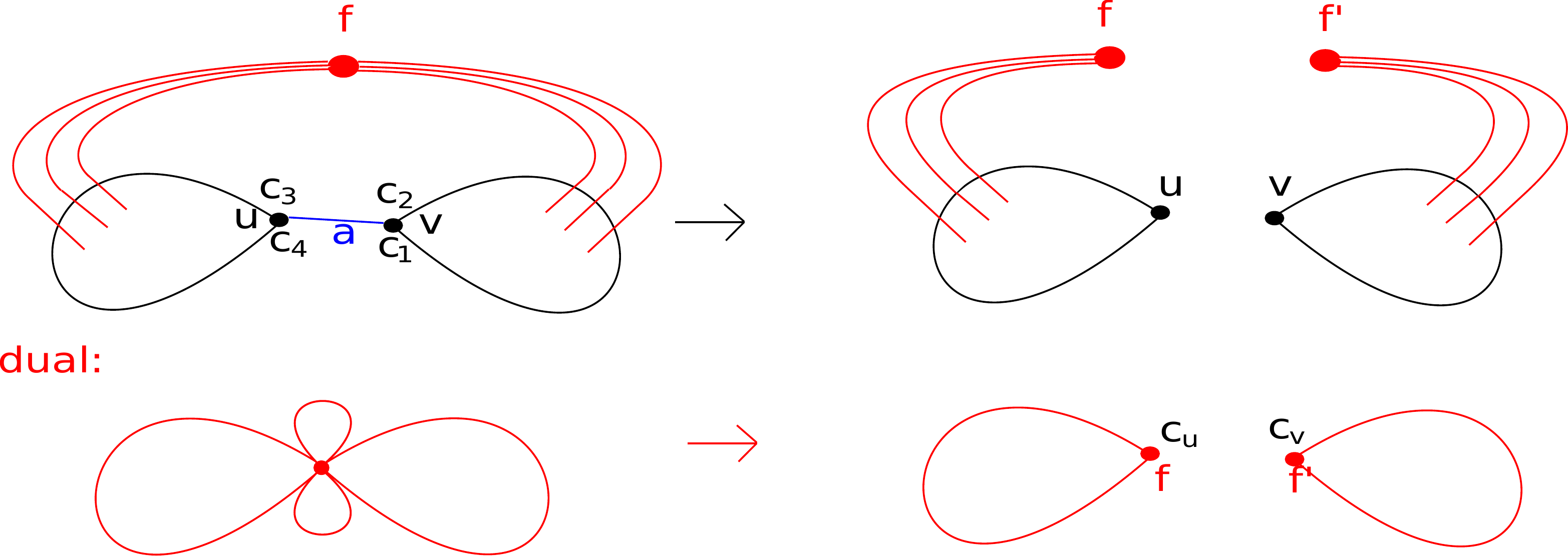}
	\caption{The bridge $(u,v)$ is deleted.}
	\label{fig:delete-bridge}
\end{figure}

An edge, $e$, is a bridge if its dual $e^\ast$ is a self-loop of some face.
To delete a bridge $e=(u,v)$ incident to the face $f$ in corners $c_1,\ldots c_4$ (see figure \ref{fig:delete-bridge}), do the following. Expose $u$ and $v$ in the primal tree. Expose the corners $c_1,\ldots c_4$ in the dual tree, meaning we have at the top a four-leaf clover of clusters incident to $f$. Cut $f$ into those four parts, and delete the leaves corresponding to $c_4 e c_1$ and $c_2 e c_3$. Let $c_u$ denote the newly formed corner incident to $u$ and $c_v$ that incident to $v$. We have now split the face $f$ in two, and have two corresponding EETs and two dual top trees. In the primal top tree, delete the edge $u$; this leaves us with two top trees, one with a leaf with boundary vertex $v$ exposed and one with boundary vertex $u$ exposed -- update the corresponding corners to be exactly the $c_v$ and $c_u$ created in the dual top tree, respectively.

If a non-bridge tree-edge is to be deleted, its dual induces a fundamental cycle in the dual tree, $C(e^\ast)$. Any edge on this cycle would reconnect the components and can be added to the primal tree as a replacement edge. For simplicity, we can just choose the edge after $e^\ast$ on that cycle.

\begin{figure}[H]
	\centering
	\includegraphics[width=0.95\linewidth]{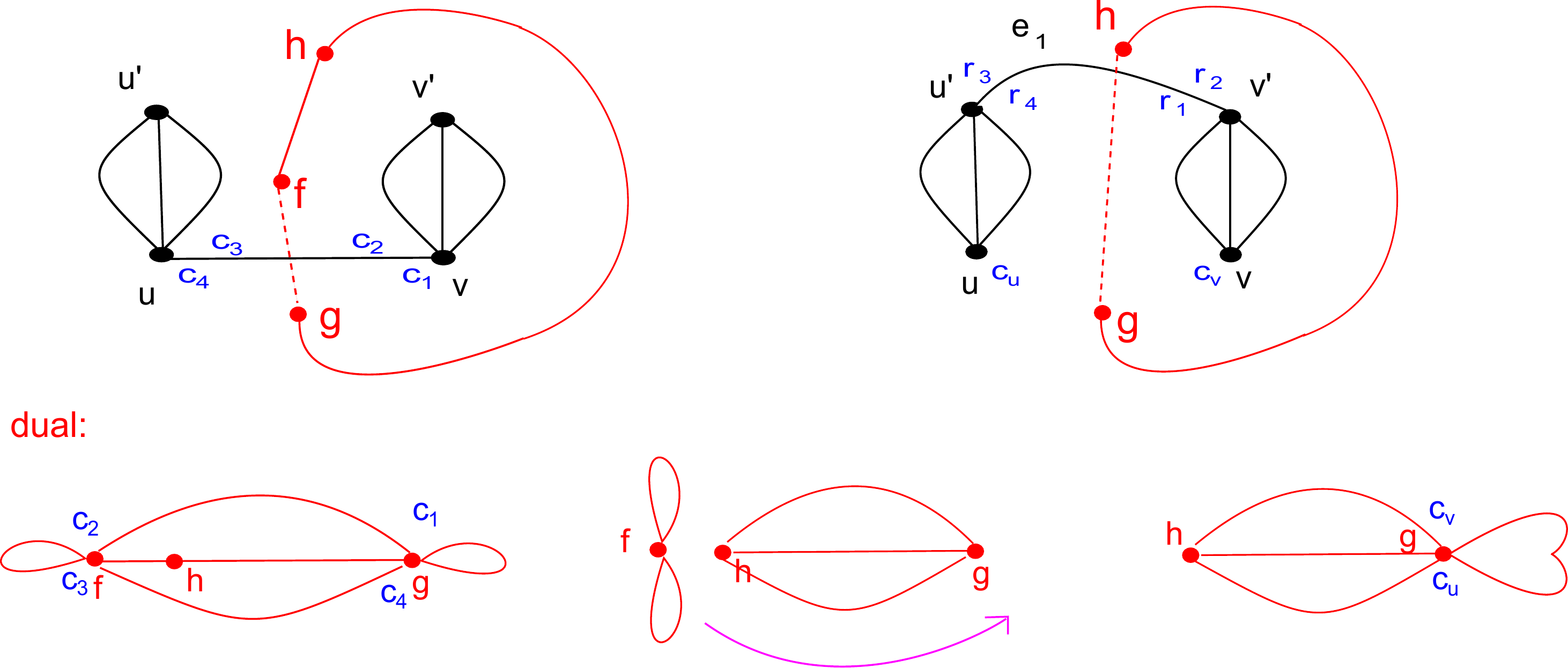}
	\caption{We delete the tree-edge $e=(u,v)$ from the graph. The non-tree edge $(u^\prime,v^\prime)$ is added to the spanning tree, and the faces on either side of $e$ are joined.}
	\label{fig:replacement-edge}
\end{figure}

Let $e=(u,v)$ be the edge we want to delete. Let $f$ and $g$ be the faces incident to $e$. First, expose $u$ and $v$ in the primal tree. The top of the primal top tree is now of the form leaf-path-leaf, where the path-cluster consists of one edge, $(u,v)$.
Since $f\neq g$, exposing corners $c_1,c_2$ incident to $e$ (and to $v$) in the dual tree returns a top tree with a path-cluster as root.  Let $a_1 = (f,h) = e_1^\ast$ be the first edge on the cluster path from $g$ to $f$. We want to replace $e$ with $e_1 = (u^\prime,v^\prime)$. (See figure \ref{fig:replacement-edge}.)

Let $c_h$ and $c_h^\prime$ be the corners of $a_1$ incident to $h$. Expose the vertices $u,v,u^\prime, v^\prime$ and the corners $c_h$ and $c_h^\prime$ in the primal top tree. 

Expose the faces $f$ and $g$ and the corners $c_1\ldots c_4$ in the dual top tree. 
In the dual tree, delete the edge $a_1=(f,h)$ from the tree. Cut $f$ in $c_2,c_3$ and $g$ in $c_1,c_4$, delete the leaves corresponding to $c_2 e c_3$ and $c_4 e c_1$, and name the replacing corners $c_f$ and $c_g$ (updating the EET). Then, join the faces $f$ and $g$ in the corners $c_f$ and $c_g$, and name the newly constructed corners $c_u$ and $c_v$. The dual top tree now has the path cluster with boundary vertices $h$ and $g$ as root.

In the primal tree, delete the edge $(u,v)$. The newly formed corners are exactly the $c_u$ and $c_v$ formed in the dual top tree. 
Link the vertices $u^\prime$ and $v^\prime$, and update the EET correspondingly. That is, if $r_1,r_2,r_3,r_4$ are corners incident to $e_1 = (u^\prime,v^\prime)$, then exchange the segments $r_2 e_1 r_1$ and $r_4 e_1 r_3$ to $r_2 e_1 r_3$ and $r_4 e_1 r_1$.


\begin{figure}[H]
	\centering
	\includegraphics[width=0.7\linewidth]{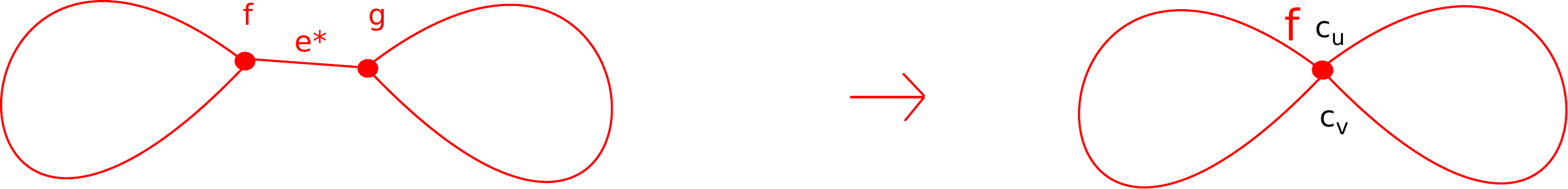}
	\caption{A non-tree edge is exposed in the dual top tree, and removed.}
	\label{fig:delete-nontree}
\end{figure}

Deleting a non-tree edge is the opposite of inserting an edge inside a face (see figure \ref{fig:connect-inside-face}).
If a non-tree edge $e=(u,v)$ is deleted, its incident faces $f$ and $g$ must be joined. In the primal tree, expose $u,v$ and the corners $c_1,\ldots ,c_4$ incident to $e$. In the dual tree, expose the edge $e^\ast = (f,g)$, which must belong to the cotree, and all four corners incident to $e$. Cut $g$ in the corners $c_1,c_4$, cut $f$ in the corners $c_2,c_3$. Delete the cluster containing $e^\ast$, and join $f$ and $g$ in the corners $c_f$ and $c_g$ which were created from the cut. The join produces new corners $c_u$ and $c_v$. In the primal tree, the leaves with EET $c_2 e c_1$ and $c_4 e c_3$ are deleted, and the delimiting corners for the cluster path are updated to be the newly formed $c_u$ and $c_v$. During this procedure, the EET is updated to contain $c_v$ in place of $c_2 e c_1$ and $c_u$ in place of $c_4 e c_3$.

\subsection{Flip}\label{sec:flip}

Finally, for $\flip$ to work we have to use a version of top trees that
is not tied to a specific clockwise orientation of the vertices.  The
version in~\cite{DBLP:journals/talg/AlstrupHLT05} that is based on a reduction to
Frederickson's topology trees~\cite{DBLP:journals/siamcomp/Frederickson85} works fine for this purpose.

\begin{definition}[Articulation flip]\label{articulationflip}
Having $\vertexsplit$ and $\vertexjoin$ functions, we may perform
an \emph{articulation-flip} (see Figure~\ref{fig:separation_flip}) --- a flip in an articulation point: 
Given a vertex $\alpha$ incident to the face $f_1$
in two corners, $c_1$ and $c_2$, we may cut through
$c_1,c_2$, obtaining two graphs $G_1,G_2$, having split $\alpha$ in
vertices $\alpha_1\in G_1, \alpha_2\in G_2$, and having introduced new
corners $c,c^{\prime}$ where we cut.  Now, given a corner $\alpha_p$
incident to $\alpha_1$ and incident to some face $f_2$, we may join $\alpha_1$ with $\alpha_2$ by the
corners $\alpha_2,\alpha_p$, 
with or without having flipped the
orientation of $G_2$.
\end{definition}

\begin{definition}[Separation flip]\label{def:separationflip}
Similarly, given a separation pair $\alpha,\beta$, incident to
the faces $f,g$ with corners $c_1,\ldots c_4$, we may cut
through those corners, obtaining two graphs.  We may then flip the
orientation of one of them, and rejoin.  We call this a
\emph{separation-flip}.
\end{definition}

\begin{SCfigure}[50]
\includegraphics[width=0.42\linewidth]{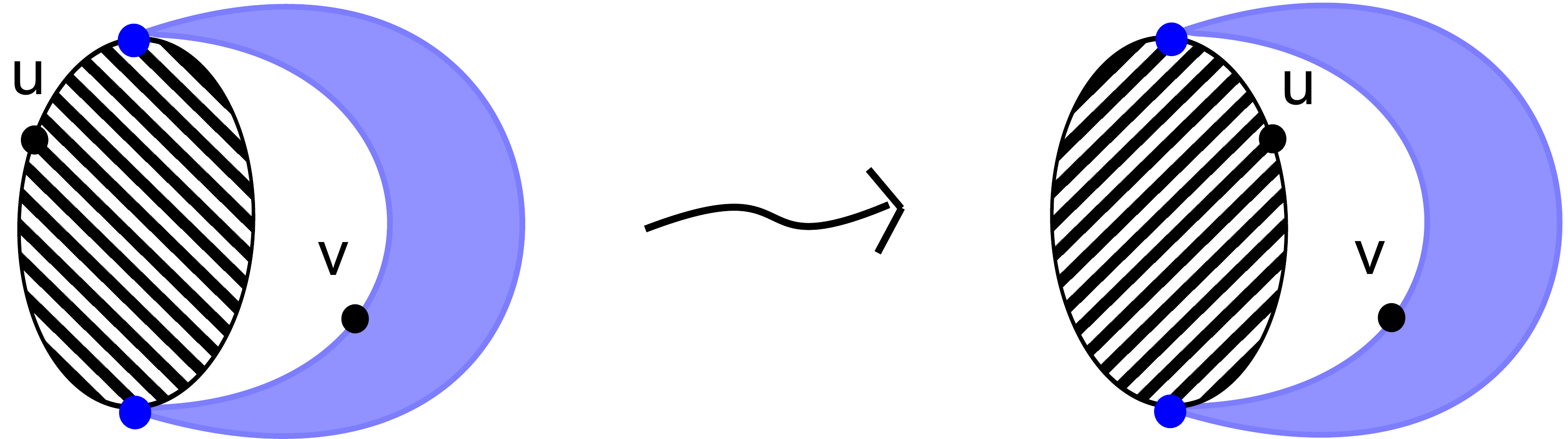}
\caption{A separation flip at a separation pair (blue). The flip makes vertex $u$ linkable with vertex~$v$.}
\label{fig:separation_flip}
\end{SCfigure}

Internally, both flip operations are done by cutting out a subgraph, altering its orientation, and joining the subgraph back in.
%
%
In the dual graph, a flip corresponds to two splits, 
two cuts,
two links, 
and two merges.





First, the articulation flip will be described. To flip in an articulation point, we only need to be able to cut a vertex, 
and join two vertices to one. (See figure \ref{fig:vertex_split}.)

\begin{figure}
	\centering
	\includegraphics[width=0.4\linewidth]{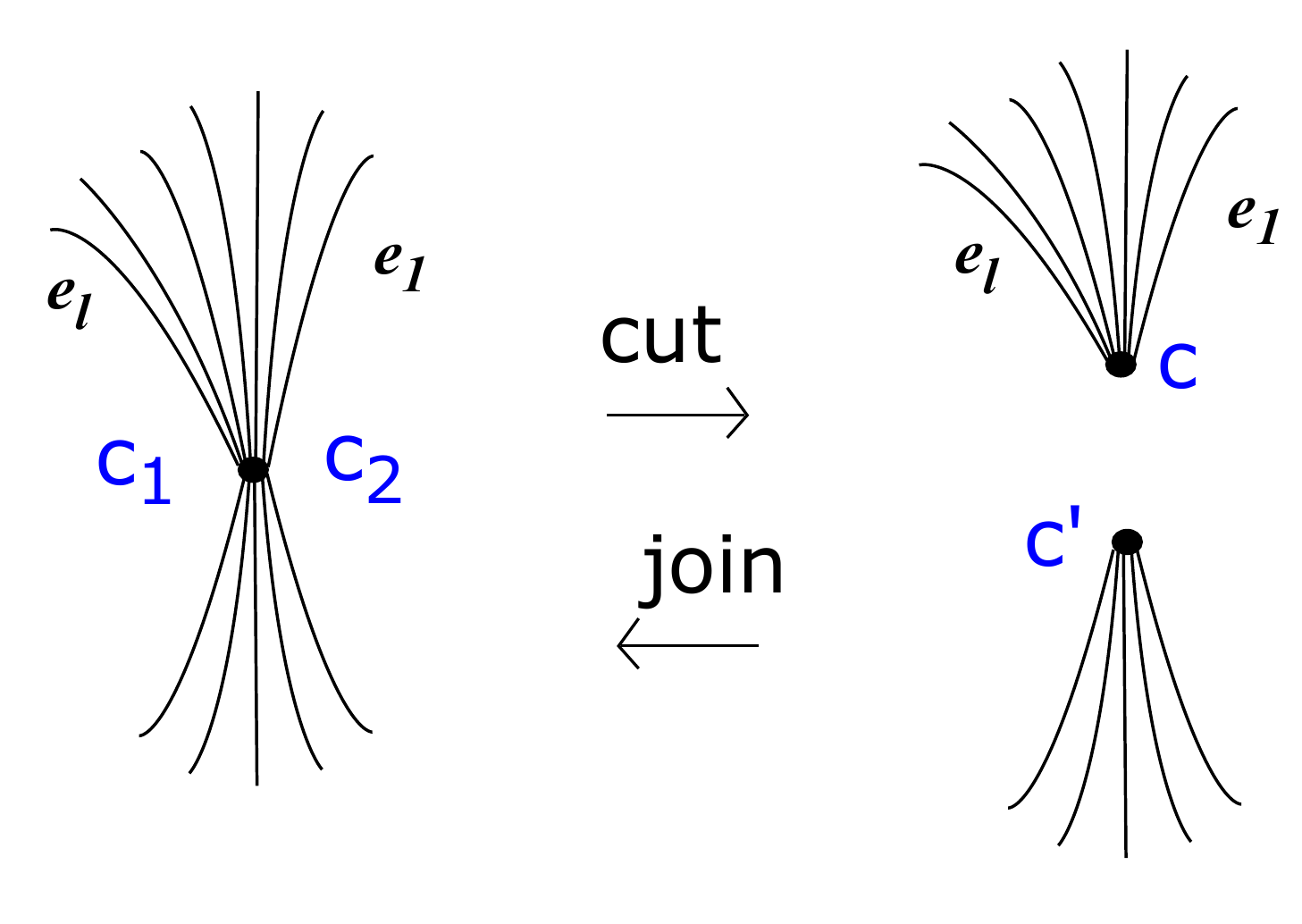}
	\caption{Left to right: The vertex is cut in two specified corners, and two new vertices are formed. Right to left: Two vertices are joined to one.}
	\label{fig:vertex_split}
\end{figure}


\subsubsection{Articulation flip}

Given a vertex $v$ and two corners $c_1,c_2$ incident to $v$, we may cut the vertex in two such that all edges between $c_1$ and $c_2$ belong to the new vertex $v^\prime$, and all edges between $c_2$ and $c_1$ belong to $v$. (See Figure \ref{fig:vertex_split}.) The corners $c_1$ and $c_2$ are deleted, and new corners $c$ and $c^\prime$ are created, and the extended Euler tours are updated correspondingly. (See figure \ref{fig:euler_tour_cut}.)


\begin{figure}[h]
	\centering
	\includegraphics[width=0.7\linewidth]{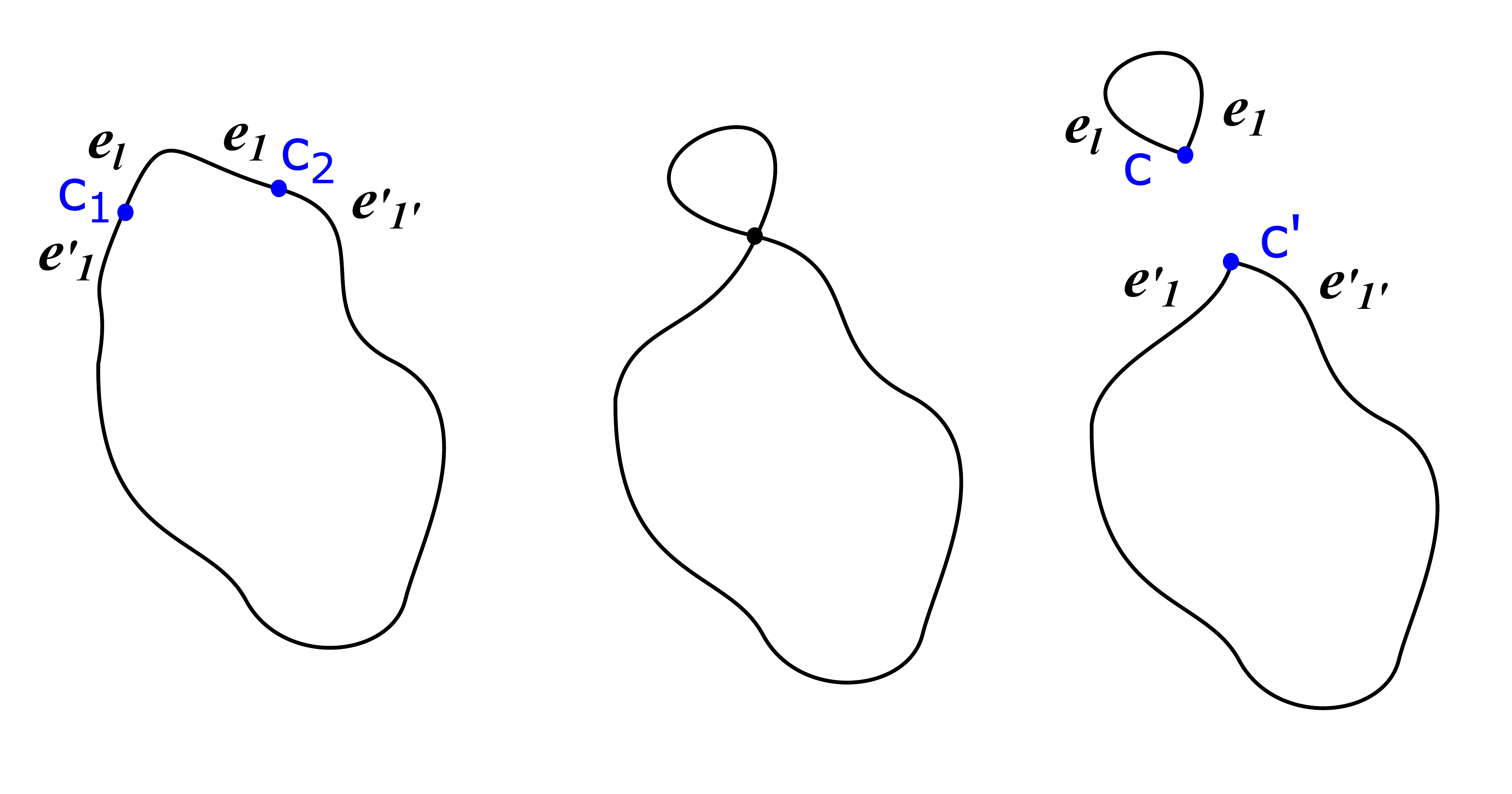}
	\caption{The extended Euler tour is cut in $c_1$ and $c_2$, and glued back together to form two tours. New corners, $c$ and $c^\prime$ are formed.}
	\label{fig:euler_tour_cut}
\end{figure}

The EET cycle is split into two cycles; let $e_1$ and $e_l$ denote the successor of $c_2$ and the predecessor of $c_1$, respectively, and let $e_1^\prime$ and $e_{l^\prime}^\prime$ denote the successor of $c_1$ and the predecessor of $c_2$, respectively. The corner $c$ will take place between $e_l$ and $e_1$ in the EET containing $v$, and the corner $c^\prime$ will take place between $e_{l^\prime}^\prime$ and $e_1^\prime$ in the EET containing $v^\prime$. (See figure \ref{fig:euler_tour_cut}.)


We need to be able to cut both primal and dual vertices. By convention, the extended Euler tour is updated when primal vertices are cut, and never when dual vertices are cut.

Similarly, given two vertices, and given a corner incident to either vertex, we may join the two vertices, as an inverse operation to the vertex cut.

If $c$ is a corner incident to the vertex $v$ and $c^\prime$ is incident to $v^\prime$, we may perform the operation join($c,c^\prime$).
%
%
Now, the extended Euler tour is updated correspondingly. We must create new corners, $c_1$ and $c_2$. In the Euler tour, $c_1$ must point to the predecessor of $c$ as its predecessor, and the successor of $c^\prime$ as its successor, and similarly for $c_2$. And vice versa, the object of those edges must point to $c_1$ and $c_2$. 

\label{sec:articulationflip}

We now have the necessary tools for flipping in an articulation point, that is, if a vertex has two corners incident to the same face, we may flip the subgraph spanned by those corners into another face incident to the vertex. This is done by the two-step scheme of cut and join. Namely:
\begin{theorem}
	We can support $\operatorname{articulation-flip(c_1,c_2,c_3)}$ in time $\Oo(\log^2 n)$.
\end{theorem}

\begin{figure}[h]
	\centering
	\includegraphics[width=0.7\linewidth]{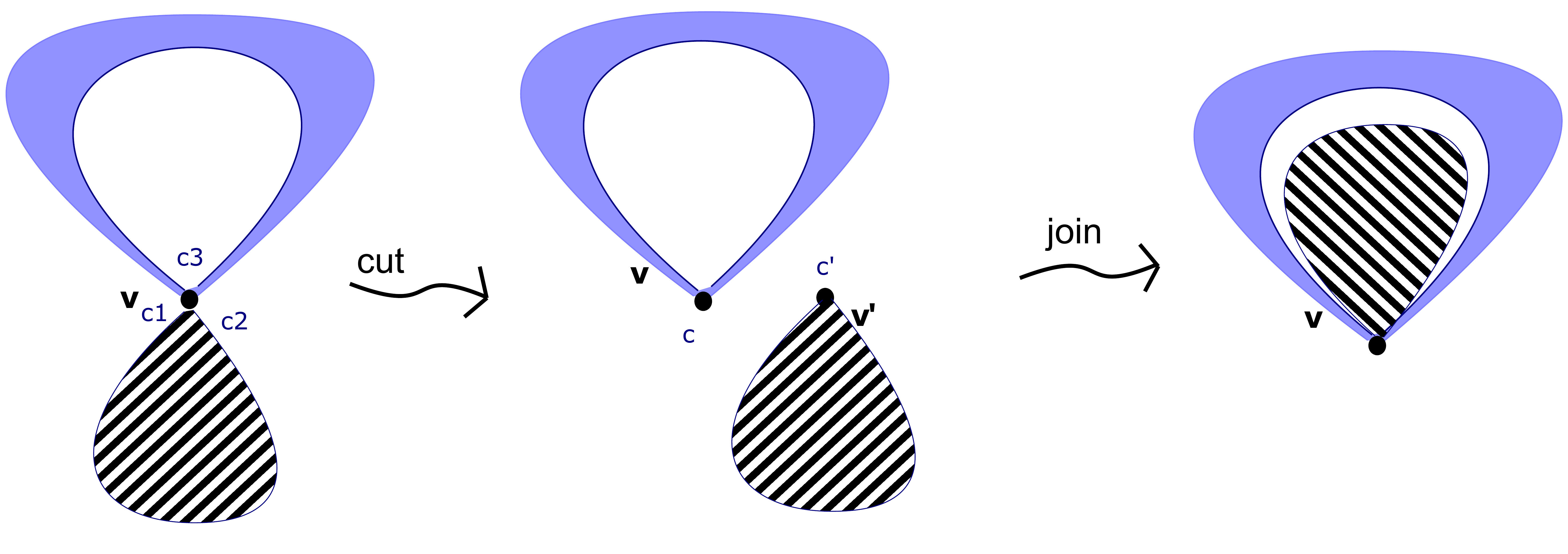}
	\caption{To flip in an articulation point: cut and join.}
\end{figure}

\begin{proof}
	Given a vertex $v$ and two corners $c_1$ and $c_2$ incident both to $v$ and some common face $f_1$, and given a third corner, $c_3$ incident to $v$ and some face $f_2$, we first perform the vertex cut, cut($v,c_1,c_2$) in the primal top tree. Then, we perform the operation of cut($f_1,c_1,c_2$) in the dual top tree. 
Together, these two operations return two graphs, $G$ and $G^\prime$, with respectively $v,c$ and $v^\prime,c^\prime$ exposed in their primal top trees, and with $f_1,c$ and $f_1^\prime,c^\prime$ exposed in their dual top trees. 
Then, we expose $c_3$ (incident to $v$ and $f_2$) in the top trees of $G$.
Finally, we perform the operation join($c_3,c^\prime$), first in the dual, and then in the primal top tree, gluing $G^\prime$ back to $G$.

This operation takes $\Oo(\log^2 n)$ time, dominated by the expose-operation in the primal top tree.\qed \end{proof}

\subsubsection{Altering the orientation}

To perform a flip in a separation pair, we introduce the following three operations: seclude, alter the orientation, include. Seclude and include are similar to those of the previous section, while altering the orientation means that the ordering of edges around any vertex becomes the opposite.
%
In the EET, the predecessor becomes a successor, and vice versa.

To implement this, we let the clusters of the primal and dual top tree contain one more piece of information, namely the orientation of the cluster, which is plus $+$, or minus $-$. 

When a cluster is ``negative'', some of its information changes character: 
\begin{itemize}
	\item its left-hand child becomes a right-hand child and vice versa, 
	\item if it is a path cluster, the two EET-segments switch places, 
	\item and, finally, in the extended Euler tour, predecessor is interpreted as successor and vice versa.
\end{itemize}

When a cluster is split, its sign is multiplied with the signs of its children. That is, if a cluster with a minus is split, the minus is propagated down to the children. When clusters are merged, the new union-cluster is simply equipped with a plus, while its children clusters keep their sign. 


In the 3-step program of ``seclude, alter, include,'' the alter-step simply consists of changing the sign of the top cluster of the dual top tree for the secluded graph.

\subsubsection{Separation flip}

Given four corners $c_1,\ldots ,c_4$ such that $c_1,c_2$ are incident to the vertex $v$, $c_3,c_4$ are incident to the vertex $u$, $c_2,c_3$ are incident to the face $f$ and $c_4,c_1$ are incident to the face $g$. See figure \ref{fig:seclude}.

\begin{figure}[h]
	\centering
	\includegraphics[width=0.3\linewidth]{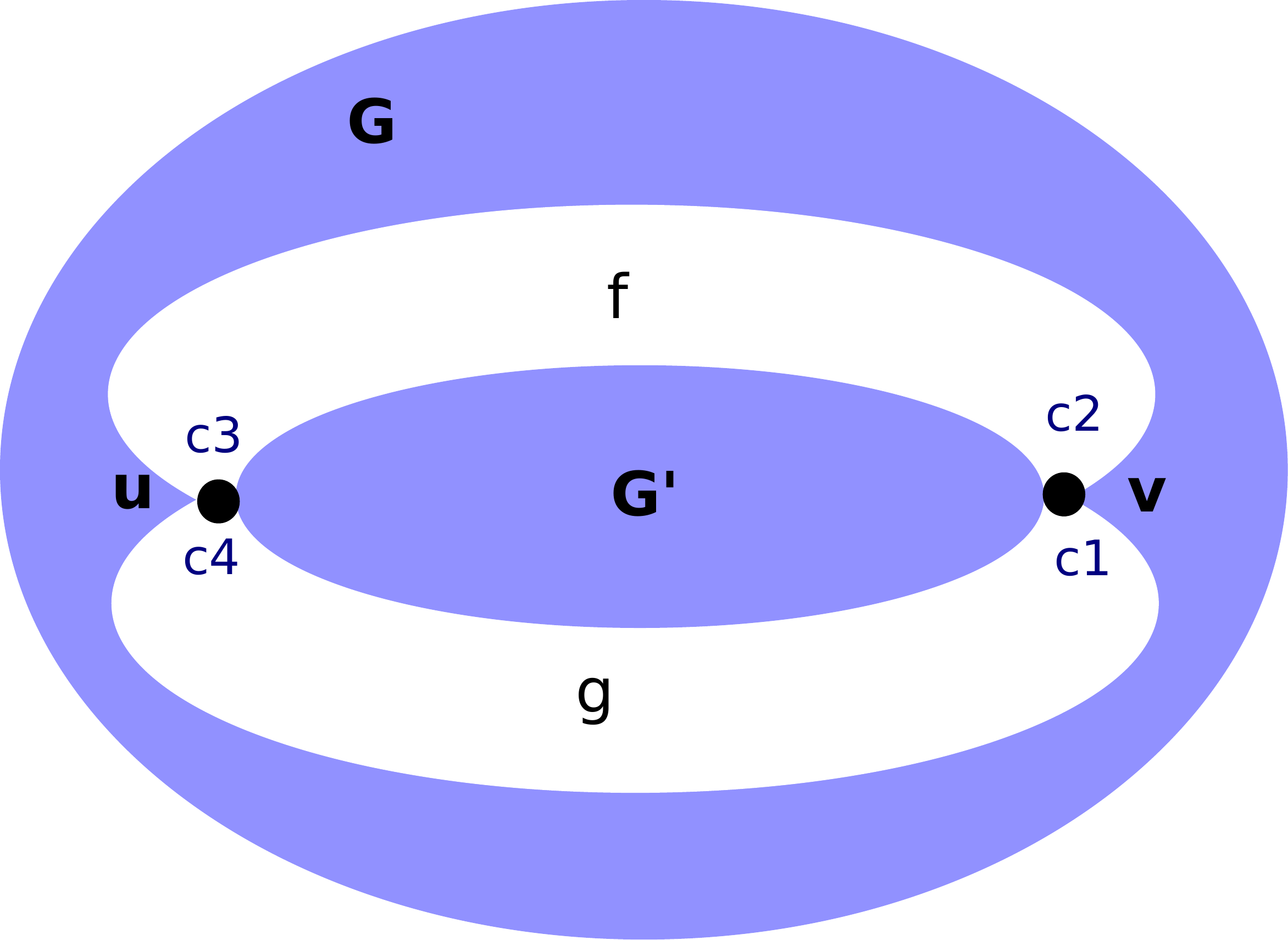}\includegraphics[width=0.1\linewidth]{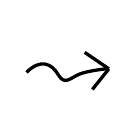}\includegraphics[width=0.6\linewidth]{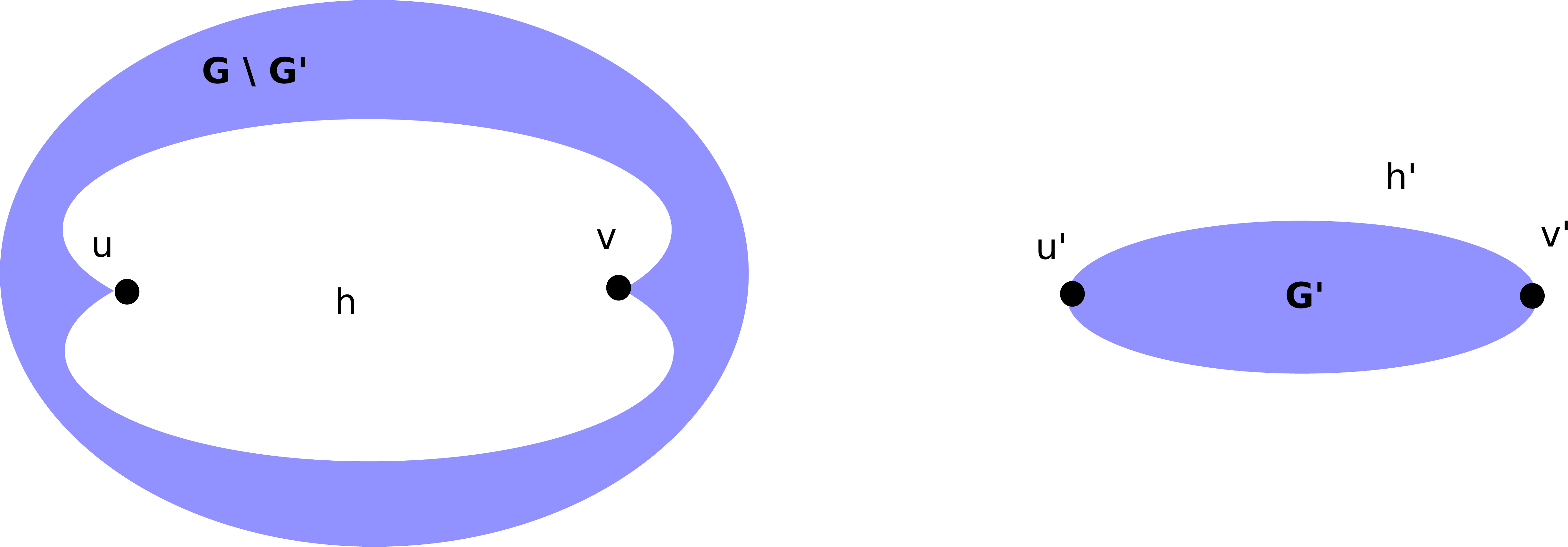}
	\caption{The subgraph $G^\prime$ is delimited by two vertices and two faces. The picture shows how $G^\prime$ is cut out of the graph $G$.}
	\label{fig:seclude}
\end{figure}


To seclude($c_1,c_2,c_3,c_4$), cut the vertex $v$ through the corners $c_1$ and $c_2$, cut $u$ through $c_3$ and $c_4$. In the dual top tree, cut the face $f$ in $c_2$ and $c_3$ (obtain $f$ and $f^\prime$), and cut the face $g$ in $c_4$ and $c_1$ (which returns $g$ and $g^\prime$). Then, join $f^\prime$ with $g^\prime$, and $f$ with $g$ in the dual top tree.
We now have two graph-stumps, denote them $G$ and $G^\prime$, where $G$ contains $u$ and $v$ which are both incident to one face $h$ which was formed by joining $f$ with $g$, and where $G^\prime$ contains $u^\prime$ and $v^\prime$, both incident to the new face $h^\prime$.

This procedure consists of two vertex cuts in the primal tree, this
takes $\Oo(\log^2 n)$ time, two vertex cuts in the dual tree, which
takes the time $\Oo(\log n)$.  Note that between the seclude and the
subsequent include we may not have a spanning tree for one of the
components.  This is not a problem, since we are not doing any other
operations in between.


Include is the inverse operation of seclude above.

Given two graphs, $G$ and $G^\prime$, and given two vertices $u,v\in G$ and two vertices $u^\prime,v^\prime\in G^\prime$, and given a designated face $h\in G$ incident to both $u$ and $v$, and similarly a face $h^\prime\in G^\prime$ incident to $u^\prime$ and $v^\prime$. Let four corners be given, such that $c_{u,h}$ is incident to $u$ and $h$, and so on, $c_{v,h}, c_{u^\prime,h^\prime}, c_{v^\prime,h^\prime}$.

To perform $\operatorname{include}(c_{u,h}, c_{v,h}, c_{u^\prime,h^\prime}, c_{v^\prime,h^\prime})$:
Cut the face $h$ through the corners $c_{v,h}$ and $c_{v,h}$, and cut the face $h^\prime$ similarly. Let the resulting faces be denoted $f,g$, and $f^\prime,g^\prime$, respectively. Now, join the faces $f$ with $f^\prime$ and $g$ with $g^\prime$. Finally, join the vertices $u$ with $u^\prime$ and $v$ with $v^\prime$. 



\begin{theorem}
	We can support $\operatorname{separation-flip}(c_1,c_2,c_3,c_4)$ in time $\Oo(\log^2 n)$.
\end{theorem}
\begin{proof}
	Given two vertices which are incident to two faces, and given four delimiting corners $c_1\ldots c_4$ pairwise incident to both vertices and both faces (as in Definition~\ref{def:separationflip}), we perform $\operatorname{separation-flip}$($c_1,c_2,c_3,c_4$) with a three-step procedure: Seclude, alter, include. 
	
	Let $g,f,c_1,\ldots,c_4$ be as in Definition~\ref{def:separationflip} (see Figure~\ref{fig:seclude}). 
	First, seclude the subgraph $G'$ delimited by the corners $c_1,\ldots,c_4$, splitting the faces $f$ and $g$ to $f,f'$ and $g,g'$, respectively. Then, alter the orientation of $G'$ by changing its sign. Finally, perform include, such that $f$ joins with $g^\prime$, and $g$ with $f^\prime$.
	
	This takes $\Oo(\log^2 n)$ time, dominated by the seclude and include operations.\qed
\end{proof}

\section{One-flip linkable query}\label{sec:flip-find}

Given vertices $u,v$, we have already presented a data structure to find a common face for $u,v$. Given they do not share a common face, we will determine if an articulation flip exists such that an edge between them can be inserted, and given no such articulation-flip exists, we will determine if a separation-flip that makes the edge insertion $(u,v)$ possible exists.

Let $f_1$ and $f_2$ be faces in $G$, and let $S$ be a subgraph of $G$.
We say that $S$ \emph{separates} $f_1$ and $f_2$ if $f_1$ and $f_2$
are not connected in $\dual{G}\setminus{}\dual{(E[S])}$. Here, $E[X]$ denotes the set of edges of the subgraph $X$, $E[f]$ the edges incident to the face $f$, and $V[f]$ the incident vertices.

\begin{wrapfigure}[14]{r}{0.35\textwidth}
\vspace{-2ex}
\includegraphics[width=1.0\linewidth]{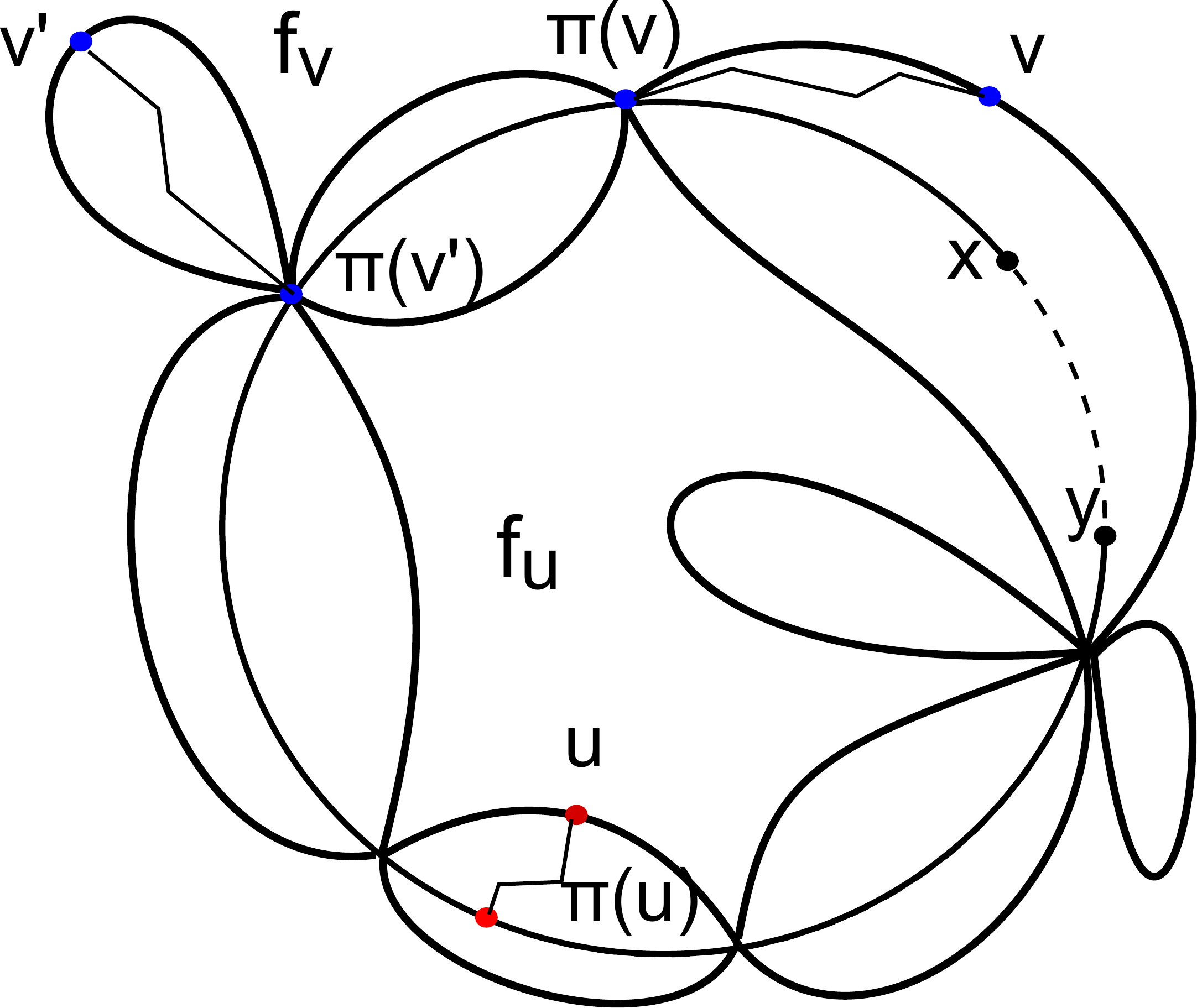}
\caption{The faces $f_u$ and $f_v$ have five common vertices, and there are eight flip-components with respect to them.} 
\label{fig:sausages}
\end{wrapfigure}

\begin{observation}\label{obs:face-separating-face-path}
  Given a fundamental cycle $C$ that is induced in $T\cup\set{e}$ by some edge $e$
  and given any two faces $f_1$, $f_2$ not separated by $C$, any
  face $f$ such that $C\cup{}E[f]$ separates $f_1$ and $f_2$ lies on the
  path $\simplepath{f_1}{f_2}$ in $\cotree{T}$.
\end{observation}

Let $f_1$ and $f_2$ be faces of $G$, and let $S=V[f_1]\cap{}V[f_2]$ be
the set of vertices they have in common. Let $C$ denote the set of corners between vertices in $S$ and faces in $\{f_1,f_2\}$. The sub-graphs obtained by cutting $G$ through all the corners of $C$ 
are called \emph{flip-components} of $G$
w.r.t. $f_1$ and $f_2$. Flip-components which are only incident to one vertex of $S$ can be flipped with an articulation-flip, and flip-components incident to two vertices can be flipped with a separation-flip. (See Figure~\ref{fig:sausages}.)

\begin{observation}\label{obs:perimeter}
Note that the perimeter of a flip component always consists of the union of a path along the face of $f_u$ with a path along the face of $f_v$. One of these paths is trivial (equal to a point) exactly when $u,v$ are linkable via an articulation-flip. 
\end{observation}

Given vertices $u$, $v$ in $G$, that are connected and not incident to a common face, we wish to find faces $f_u$ and $f_v$ such that $u$ and $v$ are in different flip-components w.r.t. $f_u$ and $f_v$.

\subsection{Finding one face}
Let $u$ and $v$ be given vertices, and assume there exist faces $f_u$
and $f_v$ such that $u\in{}V[f_u]\setminus{}V[f_v]$,
$v\in{}V[f_v]\setminus{}V[f_u]$, and $u$ and $v$ are in different
flip-components w.r.t. $f_u$ and $f_v$.

\begin{wrapfigure}[10]{r}{0.42\textwidth}
\vspace{-2em}
\includegraphics[width=1.0\linewidth]{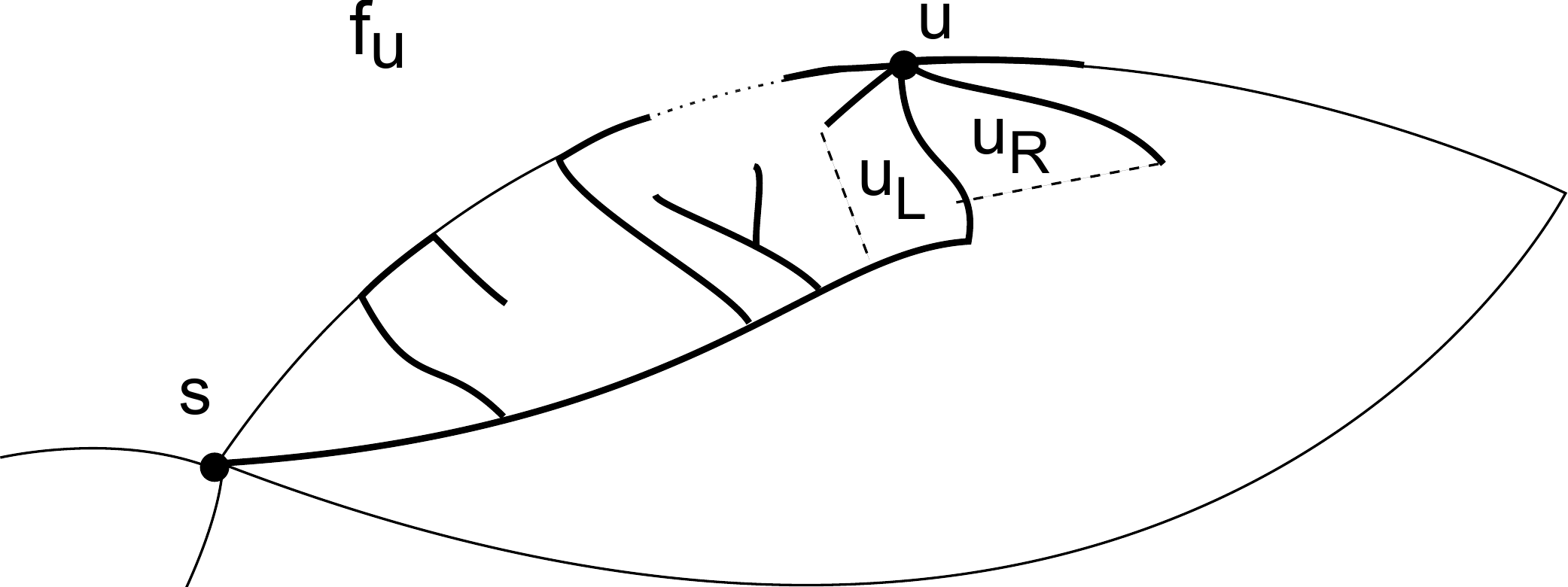}
\caption{The co-tree path from $u_L$ to $u_R$ goes through $f_u$. The proof uses that the tree-path from $u$ to $v$ goes through some $s\in S$ on the boundary of $u$'s flip-component.}
\label{fig:firstedge-cycle}
\end{wrapfigure}
Let $u_L,u_R$ be the left and right faces adjacent to the first edge
on the path from $u$ to $v$.  Similarly let $v_L,v_R$ be the left and
right faces adjacent to the first edge on the path from $v$ to $u$.

\begin{lemma}\label{lem:firstedge-cycle}
  Face $f_u$ is on the $\cotree{T}$-path $\simplepath{u_L}{u_R}$ and face $f_v$
  is on the $\cotree{T}$-path $\simplepath{v_L}{v_R}$.
\end{lemma}

\begin{proof}
For symmetry reasons, we need only be concerned with the case $f_u$.
The perimeter of a flip-component consists of edges incident to $f_u$ and edges incident to $f_v$ (see Observation~\ref{obs:perimeter}). Furthermore, in order for $u,v$ to be linkable via a flip, $u$ needs to lie on the perimeter of its flip-component. We also know that the tree-path from $v$ to $u$ must go through a point $s$ in $S$ which lies on the boundary of $u$'s flip-component. 
Thus, there must exist a path $p$ in $G$ from $\pi\in S$ to $u$, consisting only of edges incident to $f_u$. Note that $u\notin S$ since $u,v$ were not already linkable. If the first edge $e_u$ on the tree path from $u$ to $v$ is not already incident to $f_u$, then the union of $p$ and the tree must contain a fundamental cycle containing $e_u$, separating $u_L$ from $u_R$, induced by an edge $e_i$ incident to $f_u$. (See  Figure~\ref{fig:firstedge-cycle}.) But then, the co-tree path from $u_L$ to $u_R$ goes through $e_i^\ast$, which means it goes through $f_u$.
\qed
\end{proof}

\begin{lemma}\label{lem:meet-on-and-off}
  If there exists a fundamental cycle $C$ separating $f_u$ from $f_v$
  such that $u\notin{}V[C]$ and $v\in{}V[C]$, then
  $f_u=\meet(u_L,u_R,f)$ where $f\in\set{v_L,v_R}$ is the face that is
  on the same side of $C$ as $u$. Here, $\meet(a,b,c)$ denotes $a$'s projection to the path $\simplepath{b}{c}$.
  
\end{lemma}
\begin{proof}
  By Lemma~\ref{lem:firstedge-cycle} $f_u$ is on the path
  $\simplepath{u_L}{u_R}$.  And since $C\cup{}E[f_u]$ separates $f$ from
  $u_L$ and $u_R$ it is on the paths $\simplepath{u_L}{f}$ and
  $\simplepath{u_R}{f}$ by Observation~\ref{obs:face-separating-face-path}.
\qed
\end{proof}

\begin{lemma}\label{lem:meet-both-off}
  If there exists a fundamental cycle $C$ separating $f_u$ from $f_v$
  such that $u\not\in{}V[C]$ and $v\not\in{}V[C]$, then either
  $f_u=\meet(u_L,u_R,v_L)=\meet(u_L,u_R,v_R)$ or
  $f_v=\meet(v_L,v_R,u_L)=\meet(v_L,v_R,u_R)$.
\end{lemma}
\begin{proof}
  Let $e$ be the edge in $C\setminus{}T$, and let $e_u,e_v$ be the
  faces adjacent to $e$ that are on the same side of $C$ as $f_u$ and
  $f_v$ respectively.  Then $e$ is on all 4 paths in $\cotree{T}$ with
  $u_L$ or $u_R$ at one end and $v_L$ or $v_R$ at the other.
  At least one of $u,v$ is in a different flip-component from $e$, so
  we can assume without loss of generality that $u$ is.
  By Lemma~\ref{lem:firstedge-cycle} $f_u$ is on the path
  $\simplepath{u_L}{u_R}$.  And since $C\cup{}f_u$ separates $u_L$ and
  $u_R$ from $e_u$, $f_u$ is on both the paths $\simplepath{u_L}{e_u}$
  and $\simplepath{u_R}{e_u}$ by
  Observation~\ref{obs:face-separating-face-path}.
  Thus $f_u=\meet(u_L,u_R,e_u)=\meet(u_L,u_R,v_L)=\meet(u_L,u_R,v_R)$.
\qed
\end{proof}

\begin{lemma}\label{lem:meet-both-on}
  If a fundamental cycle $C$ separates $f_u$ from $f_v$
  such that $u\in{}V[C]$ and $v\in{}V[C]$, then either
  $f_u=\meet(u_L,v_L,v_R)=\meet(u_L,u_R,v_R)$ or
  $f_u=\meet(u_R,v_L,v_R)=\meet(u_L,u_R,v_L)$ or
  $f_v=\meet(v_L,u_L,u_R)=\meet(v_L,v_R,u_R)$ or
  $f_v=\meet(v_R,u_L,u_R)=\meet(v_L,v_R,u_L)$.
\end{lemma}
\begin{proof}
  Let $e$ be the edge in $C\setminus{}T$, and let $e_u,e_v$ be the
  faces adjacent to $e$ that are on the same side of $C$ as $f_u$ and
  $f_v$ respectively.  Then $e$ is on all 4 paths in $\cotree{T}$ with
  $u_L$ or $v_R$ at one end and $v_L$ or $u_R$ at the other.
  Assume that $u_L$ and $v_R$ are on
  the side of $C$ containing $f_u$ and $u_R$ and $v_L$ are on the side
  of $C$ containing $f_v$.
  At least one of $u,v$ is in a different flip-component from $e$, so
  assume that $v$ is.
  By Lemma~\ref{lem:firstedge-cycle} $f_u$ is on the path
  $\simplepath{u_L}{e_u}\subset\simplepath{u_L}{u_R}$.  And since
  $C\cup{}f_u$ separates $u_L$ and $e_u$ from $v_R$ it is on both the
  paths $\simplepath{u_L}{v_R}$ and $\simplepath{e_u}{v_R}$ by
  Observation~\ref{obs:face-separating-face-path}.
  Thus $f_u=\meet(u_L,e_u,v_R)=\meet(u_L,v_L,v_R)=\meet(u_L,u_R,v_R)$.
  The remaining cases are symmetric.
\qed
\end{proof}

\begin{theorem}
  If $f_u,f_v$ exist, either $f_u\in\set{\meet(u_L,u_R,v_L),
  \meet(u_L,u_R,v_R)}$ or $f_v\in\set{\meet(u_L,v_L,v_R),
  \meet(u_R,v_L,v_R)}$.
\end{theorem}
\begin{proof}
  If they exist there is at least one fundamental cycle $C$ separating
  them.  This cycle must have the properties of at least one of
  Lemmas~\ref{lem:meet-on-and-off}, \ref{lem:meet-both-off}, or~\ref{lem:meet-both-on}.
\qed
\end{proof}

By computing the at most two different $\meet$ values and checking
which ones (if any) contain $u$ or $v$ we therefore get at most two
candidates and are guaranteed that at least one of them is in
$\set{f_u,f_v}$ if they exist.



\subsection{Finding the other face}

\begin{lemma}
  Let $u$, $v$, and $f_u$ be given. Then the first edge $e_L$ on 
  $\simplepath{f_u}{v_L}$ or the first edge $e_R$ on $\simplepath{f_u}{v_R}$ induces a fundamental cycle
  $C(e_R)$ or $C(e_L)$ in $T$ 
  that separates $f_u$ from $f_v$.
\end{lemma}
\begin{proof}
  By lemma~\ref{lem:firstedge-cycle}, $f_v$ is on
  $\simplepath{v_L}{v_R}$ in $\cotree{T}$, so the first edge on
  $\simplepath{f_u}{f_v}$ is also the first edge on either
  $\simplepath{f_u}{v_L}$ or $\simplepath{f_u}{v_R}$.
\qed
\end{proof}

Thus given the correct $f_u$ we can find at most two candidates for an edge $e$ that induces a fundamental cycle $C(e)$ in $T$ that separates $f_u$ from $f_v$, and be guaranteed that 
one of them is correct.

\begin{observation}\label{obs:pies}
For each vertex, $v$, we may consider the projection $\pi(v)$ of $v$ onto the cycle $C$.
For each flip-component, $X$, we may consider the projection $\pi(X)=\{\pi(v)\mid v\in X\}$. If $X$ is an articulation-flip component, the projection $\pi(X)$ is a single point in $S=V\left[f_u\right]\cap V\left[f_v\right]$. If $X$ is a separation-flip component, its projection is a segment of the cycle, $\simplepath{\pi_1}{\pi_2}$, between the separation pair $(\pi_1,\pi_2)\subset C(e)$ where $\pi_1,\pi_2\in S$.
\end{observation}

\subsubsection{Finding an articulation-flip.}
Let $(x,y)$ be any edge inducing a cycle $C$ in $T\cup\set{(x,y)}$
that separates $f_u$ from $f_v$, let $\pi(u)=\meet_T(u,x,y)$ be the projection of $u$ on $C$. 

Now the articulation-flip cases are not necessarily symmetrical. 
First we present how to detect an articulation-flip, given $u,v,$ and $f_u$, if $f_v$ plays the role of $f_2$ (see Section~\ref{sec:articulationflip}).

If the flip-component containing $v$ is an articulation-flip component, then $\pi(v)$ is an articulation point incident to both $f_u$ and $f_v$, but the opposite is not necessarily the case. Assume $\pi(v)$ is incident to both $f_u$ and $f_v$ and let $c_u$ denote a corner between $\pi(v)$ and $f_u$.

\begin{SCfigure}[50]
\centering
\caption{ If $f_v$ is an articulation point, so is $\pi(v)$. But then the co-tree path from $u_L$ to $v_L$ must go through $f_v$.
Left: Primal graph. Right: Dual graph.
}
\includegraphics[width=0.6\textwidth]{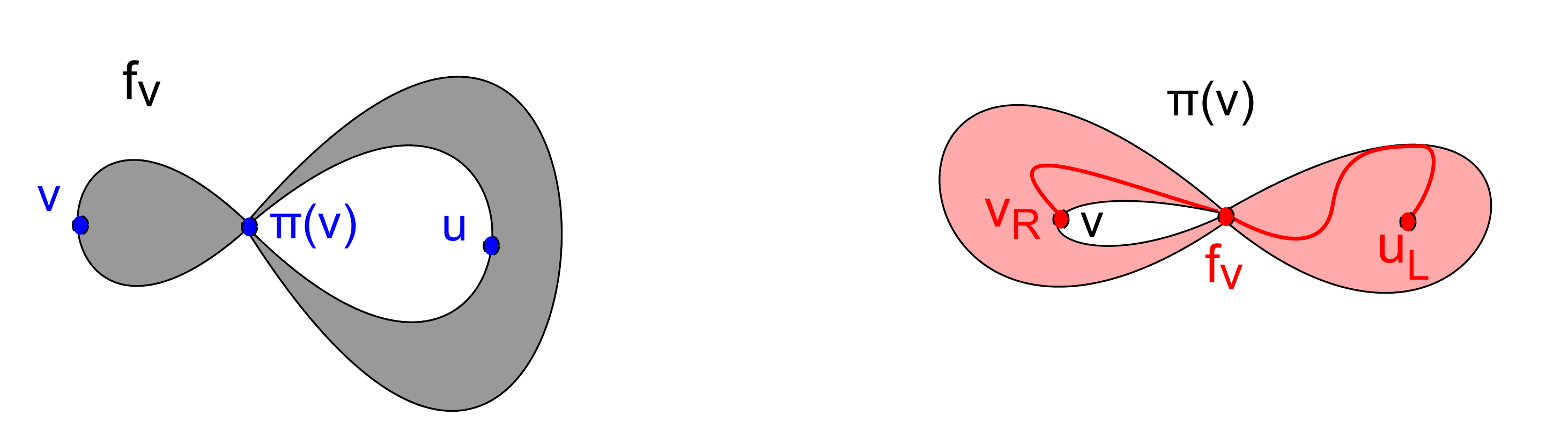}
\label{fig:articulation_dual}
\end{SCfigure}

Note that if $\pi(v)$ is an articulation point with corners $c_1,c_2$ both incident to $f_v$, then $f_v$ is an articulation point in the dual graph with corners $c_1,c_2$ both incident to $\pi(v)$. Removing $f_v$ from the dual graph would split its component into several components, and clearly, aside from $f_v$, only faces in \emph{one} of these components may contain faces incident to $v$. Any path in the co-tree starting and ending in different components w.r.t the split will have the property that the first face incident to $v$ on that path is $f_v$. (See Figure~\ref{fig:articulation_dual}.)

Now, in the case $f_u=f_1$ and $f_v=f_2$, to find the corner of $\pi(v)$ incident to $f_u$, we can simply use $\operatorname{linkable}(\pi(v),u)$ from before, which will return a corner of $f_u$ incident to $\pi(v)$. 
To find the two corners of $f_v$: 
With the dual structure (see Observation~\ref{obs:dualstruct}) 
we may mark the face $f_v$, and expose the vertices $u,v$. Now, $\pi(v)$ has a unique place in the face-list of some cluster --- if and only if that place is in the root cluster, and $n_{\min} = s_{\min} = 0$ for 
that cluster, $f_v$ plays the role of $f_2$. That is, if and only if $\pi(v)$ has a corner incident to $f_v$ to one side, and a corner incident to $f_v$ to the other side. In affirmative case, $\pi(v)$ appears with at least one corner to either 
vertex list; those corners can now be used as cutting-corners for the articulation-flip.

If instead $f_v$ played the role of $f_1$, a similar procedure is done with $\pi(u)$.

\begin{theorem}
Given $u,v$ are not already linkable, we can determine whether $u,v$ are linkable via an articulation-flip in time $\Oo(\log^2 n)$.
\end{theorem}

\subsubsection{Finding a separation-flip}

Assume $v,u$ are not linkable via an articulation-flip, determine if they are linkable via a separation-flip.

\begin{lemma}
Let $(x,y)$ be any edge inducing a cycle $C$ in $T\cup\set{(x,y)}$
that separates $f_u$ from $f_v$, let $\pi(u)=\meet_T(u,x,y)$ be the projection of $u$ on $C$. 
  Let $e_1, e_2$ be the edges incident to $\pi(u)$ on $C$.  Then at least
  one of $e_1$, $e_2$ is in the same flip-component as $u$ w.r.t $f_u$
  and $f_v$.
\end{lemma}
\begin{proof}
This follows from Observation~\ref{obs:pies}: If $\pi(u)$ is 
incident to both $f_u$ and $f_v$, then exactly one of the edges $e_1,e_2$ is in the same flip-component as $u$. Otherwise,
both of the edges $e_1,e_2$ are in the same flip-component as $u$.
\qed
\end{proof}

\begin{lemma}
  Let $C$ be any fundamental cycle separating $f_u$ from $f_v$, let $e_u$
  be an edge on $C$ in the same flip-component as $u$, let $f_1$ be
  the face adjacent to $e_u$ that is separated from $f_u$ by $C$, and
  let $f_2\in\set{v_L,v_R}$ be a face on the same side of $C$ as
  $f_1$.  Then $f_v$ is the first face on $\simplepath{f_1}{f_2}$ that
  contains $v$.
\end{lemma}
\begin{proof}
  $C\cup{}E[f_v]$ separates $f_1$ and $f_2$, so by
  Observation~\ref{obs:face-separating-face-path} $f_v$ is on the
  $\simplepath{f_1}{f_2}$ path.  It must be the first face on that
  path that contains $v$ because for any face $f$ after that,
  $C\cup{}V[f]$ does not separate $f_1$ and $f_2$, since it can only
  touch the part of $C$ between $u'=\meet(u,x,y)$ and
  $v'=\meet(v,x,y)$ where $(x,y)$ is the edge inducing $C$.
\qed
\end{proof}

\subsection{Finding the separation pair and corners}

Assume $u,v$ are not linkable and not linkable via an articulation-flip.

\begin{lemma}
  Given $u$, $v$, $f_u$, and $f_v$, let $\dual{(x,y)}$ be any edge on
  $\simplepath{f_u}{f_v}$ inducing a separating cycle $C$.
  If $\pi(u)=\pi(v)=\alpha$, then $\alpha$ is one of the separation
  points if it is adjacent to both $f_u$ and $f_v$, and otherwise no
  separation pair for $u,v$ exists.  The other separation point,
  $\beta$, is then the first vertex $\neq \alpha$ adjacent to both $f_u$ and $f_v$
  on either $\simplepath{\alpha}{x}$ or $\simplepath{\alpha}{y}$.
  If instead $\pi(u)\neq\pi(v)$, then $\alpha,\beta$ are among the first two
  vertices adjacent to both $f_u$ and $f_v$ either on
  $\simplepath{\pi(u)}{x}$ and $\simplepath{v}{x}$, or on
  $\simplepath{u}{y}$ and $\simplepath{v}{y}$.
\end{lemma}
\begin{proof}
If the projection of $u$ equals the projection of $v$, but $u$ and $v$ are in different flip-components, then 
the next point incident to both $f_u$ and $f_v$ along the cycle to either side will be the one we are looking for. However, $(x,y)$ may be internal in the flip component containing $u$ or that containing $v$, and thus one of the searches may return the empty list. But then the other will return the desired pair of vertices.

If the projections are different, and do not themselves form the desired pair $(\alpha,\beta)$, then we may assume
without loss of generality that $\pi(v)$ does not belong to the flip-component containing $u$. Let $X_v,X_u$ denote the flip-components containing $u$ and $v$, respectively. If $(x,y)$ is in $X_v$, such that no edge on $\simplepath{\pi(v)}{x}$ is incident to both $f_u$ and $f_v$, then the first vertex on $\simplepath{\pi(v)}{y}$ incident to $f_u$ and $f_v$ is $\alpha$. Recall (Observation~\ref{obs:pies}) that $\pi(X_u)$ is an arc $\pi_1 \cdots \pi_2\subset C$, and suppose without loss of generality  $\pi_1$ is on the path $\simplepath{u}{v}$. If $\pi(u) = \pi_1$, $\beta$ is the second vertex on the path $u$ to $y$ incident to both $f_u$ and $f_v$, as $\pi(u)$ itself is the first. Otherwise, the first such vertex on the path is $\beta$. 
If, on the other hand, $(x,y)$ did not belong to $X_v$, let $x$ be the vertex of $x,y$ with the property that the path $\simplepath{u}{x}$ goes through $\pi(u)$. Then the first vertices on the paths to $x$ which are incident to $f_u$ and $f_v$ both, will be the desired separation pair.
\qed
\end{proof}

\begin{lemma}
In the scenario above, we may find the first two vertices on the path incident to both faces in time $\Oo(\log^2 n)$.
\end{lemma}
 
\begin{proof}
We use the dual structure (see Observation~\ref{obs:dualstruct}) to 
search for vertices incident to $f_u$ and $f_v$. Now since the path $\simplepath{\pi(u)}{x}$ is a sub-path of the fundamental cycle $C$ induced by $(x,y)$ which separates $f_u$ from $f_v$, all corners incident to $f_v$ will be on one side,
and all corners incident to $f_u$ will be on the other side
of the path, or at the endpoints. Thus, we expose $f_u$ and $f_v$ in the dual structure, which takes time $\Oo(\log^2 n)$. Now expose $\pi(u),x$ in the primal tree. Since this path is part of the separating cycle, if $n_{\min}=s_{\min}=0$, 
then the maintained vertex-list will contain exactly those vertices incident to both faces, and a corner list for each of them. 
We now deal separately with the endpoints exactly as with $\linkablequery$, by exposing the endpoint faces one by one in the dual structure, and noting whether $c_{\min}=0$ and in that case, the corner list, for each endpoint. 
\qed
\end{proof}
We conclude with the following theorem.
\begin{theorem}
  We can maintain an embedding of a dynamic graph under $\insertrm$ (edge),
  $\remove$ (edge), $\vertexsplit$ (vertex), $\vertexjoin$ (vertex), and $\flip$ (subgraph), together with
  queries that
  \begin{enumerate}
  \item \emph{(linkable)} Answer whether an edge can be inserted
    between given endpoints with no other changes to the embedding, and if
    so, where.
  \item \emph{(one-flip-linkable)} Answer whether there exists a flip that would change the
    answer for \emph{linkable} from ``no'' to ``yes'', and
    if so, what flip.
  \end{enumerate}
  The worst case time per operation is $\Oo(\log^2 n)$.
  
\end{theorem}

\subsubsection*{Acknowledgments}

We would like to thank Christian Wulff-Nilsen and Mikkel Thorup for
many helpful and interesting discussions and ideas.

\nocite{Sleator1983362}
\nocite{Tutte1963}
\nocite{Klein:2005}

\newpage
\bibliographystyle{plain}
\bibliography{refs}
\end{document}